\documentclass[sigconf]{acmart}

\usepackage{citesort}
\usepackage{color}
\usepackage[ruled, vlined, linesnumbered]{algorithm2e}
\usepackage{amsmath,amssymb}
\usepackage{caption}
\DeclareCaptionType{copyrightbox}
\usepackage{framed}
\usepackage{tikz}
\usepackage{multirow}
\usepackage{bbm}
\usepackage{float}
\usepackage{afterpage}
\usepackage{epsfig}
\usepackage{epstopdf}
\usepackage{subfig}

\allowdisplaybreaks

\usepackage{graphicx}
\usepackage{subfig}
\usepackage{float}
\usepackage{afterpage}
\usepackage{color}
\usepackage{bbm}
\usepackage{epsfig}
\usepackage{epstopdf}
\usepackage{balance}
\allowdisplaybreaks

\newtheorem{observation}{Observation}

\newtheorem{remark}{Remark}
\newtheorem{fact}{Fact}

\newcommand{\eps}{\epsilon}
\newcommand{\abs}[1]{\left| #1 \right|}

\newcommand{\bE}{\mathbf{E}}

\renewcommand{\Pr}{\mathbf{Pr}}

\renewcommand{\S}{\mathcal{S}}
\renewcommand{\G}{\mathcal{G}}
\newcommand{\Grid}{\mathbb{G}}
\newcommand{\ball}{\mathtt{Ball}}

\newcommand{\Sa}{S^{\text{acc}}}
\newcommand{\Sr}{S^{\text{rej}}}
\newcommand{\Sf}{S^{\text{rep}}}

\newcommand{\adj}{\textsc{adj}}
\newcommand{\cell}{\textsc{cell}}

\renewcommand{\S}{\mathcal{S}}

\newcommand{\alg}{\mathtt{ALG}}

\newcommand{\opt}{n_{\mathtt{opt}}}
\newcommand{\sol}{n_{\mathtt{gdy}}}
\newcommand{\Gsub}{\mathcal{G}_\mathtt{sub}}

\renewcommand{\paragraph}[1]{\medskip \noindent {\bf #1.}}

\begin{document}

\copyrightyear{2018}
\acmYear{2018}
\setcopyright{acmlicensed}
\acmConference[PODS'18]{37th ACM SIGMOD-SIGACT-SIGAI
Symposium on Principles of Database Systems}{June 10--15,
2018}{Houston, TX, USA}
\acmBooktitle{PODS'18: 37th ACM SIGMOD-SIGACT-SIGAI Symposium
on Principles of Database Systems, June 10--15, 2018, Houston, TX,
USA}
\acmPrice{15.00}
\acmDOI{10.1145/3196959.3196978}
\acmISBN{978-1-4503-4706-8/18/06}

\fancyhead{}


\title{Distinct Sampling on Streaming Data with Near-Duplicates}

\titlenote{Both authors are supported by NSF CCF-1525024 and IIS-1633215.}

\author{Jiecao Chen}
\affiliation{%
  \institution{Indiana University Bloomington}
  \city{Bloomington} 
  \state{IN} 
  \postcode{47408}
  \country{USA}
}
\email{jiecchen@umail.iu.edu}

\author{Qin Zhang}
\affiliation{%
  \institution{Indiana University Bloomington}
  \city{Bloomington} 
  \state{IN} 
  \postcode{47408}
  \country{USA}
}
\email{qzhangcs@indiana.edu}

\begin{abstract}
In this paper we study how to perform distinct sampling in the streaming model where data contain near-duplicates.  The goal of distinct sampling is to return a distinct element uniformly at random from the universe of elements, given that all the near-duplicates are treated as the same element.  We also extend the result to the sliding window cases in which we are only interested in the most recent items.  We present algorithms with provable theoretical guarantees for datasets in the Euclidean space, and also verify their effectiveness via an extensive set of experiments.
\end{abstract}

\maketitle

\section{Introduction}
\label{sec:intro}

Real world datasets are always noisy;  imprecise references to same real-world entities are ubiquitous in the business and scientific databases.  For example, YouTube contains many videos of almost the same content; they appear to be slightly different due to cuts, compression and change of resolutions.  A large number of webpages on the Internet are near-duplicates of each other.  Numerous tweets and WhatsApp/WeChat messages are re-sent with small edits.  This phenomenon makes data analytics more difficult.  It is clear that direct statistical analysis on such noisy datasets will be erroneous.
For instance, if we perform standard distinct sampling, then the sampling will be biased towards those elements that have a large number of near-duplicates.

On the other hand, due to the sheer size of the data it becomes infeasible to perform a comprehensive data cleaning step before the actual analytic phase.  In this paper we study how to process datasets containing near-duplicates in the data stream model~\cite{FM85,AMS99}, where we can only make a sequential scan of data items using a small memory space before the query-answering phase.  When answering queries we need to treat all the near-duplicates as the same universe element.

This general problem has been recently proposed in \cite{CZ16}, where the authors studied the estimation of the number of distinct elements of the data stream (also called $F_0$).  In this paper we extend this line of research by studying another fundamental problem in the data stream literature: the {\em distinct sampling} (a.k.a.\ $\ell_0$-sampling), where at the time of query we need to output a random sample among all the distinct elements of the dataset.  $\ell_0$-sampling has many applications that we shall mention shortly.

We remark, as also pointed out in \cite{CZ16}, that we cannot place our hope on a {\em magic hash function} that can map all the near-duplicates into the same element and otherwise into different elements, simply because such a magic hash function, if exists, needs a lot of bits to describe.

\paragraph{The Noisy Data Model and Problems}
Let us formally define the noisy data model and the problems we shall study.  In this paper we will focus on points in the Euclidean space.  More complicated data objects such as documents and images can be mapped to points in their feature spaces. 

We first introduce a few concepts (first introduced in \cite{CZ16}) to facilitate our discussion.  Let $d(\cdot, \cdot)$ be the distance function of the Euclidean space, and let $\alpha$ be a parameter (distance threshold) representing the maximum distance between any two points in the same group.  

\begin{definition}[data sparsity]
\label{def:sparse}
We say a dataset $S$ {\em $(\alpha, \beta)$-sparse} in the Euclidean space for some $\beta \ge \alpha$  if for any $u, v \in S$ we have either $d(u, v) \le \alpha$ or $d(u, v) > \beta$.  We call $\max_\beta \beta/\alpha$ the {\em separation ratio}.
\end{definition}

\begin{definition}[well-separated dataset]
\label{def:well-separated} 
We say a dataset $S$ {\em well-separated} if the separation ratio of $S$ is larger than $2$.
\end{definition}

\begin{definition}[natural partition; $F_0$ of well-separated dataset]
We can naturally partition a well-separated dataset $S$ to a set of groups such that the intra-group distance is at most $\alpha$, and the inter-group distance is more than $2\alpha$. We call this the unique {\em natural partition} of $S$.  Define the number of distinct elements of a well-separated dataset w.r.t.\ $\alpha$, denoted as $F_0(S, \alpha)$, to be the number of groups in the natural partition.  
\end{definition}
We will assume that $\alpha$ is given as a user-chosen input to our algorithms. In practice, $\alpha$ can be obtained for example by sampling a small number of items of the dataset and then comparing their labels.  

For a general dataset, we need to define the number of distinct elements as an optimization problem as follows.
\begin{definition}[$F_0$ of general dataset]
\label{def:F0}
The number of distinct elements of $S$ given a distance threshold $\alpha$,  denoted by $F_0(S, \alpha)$, is defined to be the size of the {\em minimum} cardinality partition $\G = \{G_1, G_2, \ldots, G_n\}$ of $S$ such that for any $i = 1, \ldots, n$, and for any pair of points $u, v \in G_i$, we have $d(u, v) \le \alpha$.
\end{definition}

Note that the definition for general datasets is consistent with the one for well-separated datasets. 

We next define $\ell_0$-sampling for noisy datasets. To differentiate with the standard $\ell_0$-sampling we will call it {\em robust} $\ell_0$-sampling; but we may omit the word ``robust'' in the rest of the paper when it is clear from the context.  We start with well-separated datasets.

\begin{definition}[robust $\ell_0$-sampling on well-separated dataset]
\label{def:sampling}
Let $S$ be a well-separated dataset with natural partition $\G = \{G_1, G_2,$ $\ldots, G_n\}$.
The robust $\ell_0$-sampling on $S$ outputs a point $u \in S$ such that
\begin{equation}
\label{eq:def-1}
\forall i \in [n], \Pr[u \in G_i] = 1/n.
\end{equation} 
That is, we output a point from each group with equal probability; we call the outputted point the {\em robust $\ell_0$-sample}.
\end{definition}

It is a little more subtle to define robust $\ell_0$-sampling on general datasets, since there could be multiple minimum cardinality partitions, and without fixing a particular partition we cannot define $\ell_0$-sampling.  We will circumvent this issue by targeting a slightly weaker sampling goal.
\begin{definition}[robust $\ell_0$-sampling on general dataset]
Let $S$ be a dataset and let $n = F_0(S, \alpha)$. The robust $\ell_0$-sampling on $S$ outputs a point $q$ such that, 
\begin{equation} 
\label{eq:def-2}
\forall p \in S, \Pr[q \in \ball(p, \alpha) \cap S] = \Theta(1/n),
\end{equation}
where $\ball(p, \alpha)$ is the ball centered at $p$ with radius $\alpha$.  
\end{definition}

Let us compare Equation (\ref{eq:def-1}) and (\ref{eq:def-2}). It is easy to see that when $S$ is well-separated, letting $G(p)$ denote the group that $p$ belongs to in the natural partition of $S$, we have
\begin{equation*}
G(p) = \ball(p, \alpha) \cap S,
\end{equation*}
and thus we can rewrite (\ref{eq:def-1}) as
\begin{equation} 
\label{eq:def-3}
\forall p \in S, \Pr[q \in \ball(p, \alpha) \cap S] = 1/n.
\end{equation}
Comparing (\ref{eq:def-2}) and (\ref{eq:def-3}), one can see that the definition of robust $\ell_0$-sampling on general dataset is consistent with that on well-separated dataset, except that we have relaxed the sample probability by a constant factor.

\paragraph{Computational Models} We study robust $\ell_0$-sampling in the standard streaming model, where the points $p_1, \ldots, p_m \in S$ comes one by one in order, and we need maintain a sketch of $S_t = \{p_1, \ldots, p_t\}$ (denoted by $sk(S_t)$)  such that at any time $t$ we can output an $\ell_0$-sample of $S_t$ using $sk(S_t)$. The goal is to minimize the size of sketch $sk(S_t)$ (or, the memory space usage) and the processing time per point under certain accuracy/approximation guarantees.

We also study the {\em sliding window} models. Let $w$ be the window size. In the {\em sequence-based} sliding window model, at any time step $t$ we should be able to output an $\ell_0$-sample of $\{p_{\ell-w+1}, \ldots, p_\ell\}$ where $p_\ell$ is the latest point that we receive by the time $t$.  In the {\em time-based} sliding window model, we should be able to output an $\ell_0$-sample of $\{p_{\ell'}, \ldots, p_\ell\}$ where $p_{\ell'}, \ldots, p_\ell$ are points received in the last $w$ time steps $t -w+1, \ldots, t$.  The sliding window models are generalizations of the standard streaming model (which we call the {\em infinite window} model), and are very useful in the case that we are only interested in the most recent items.   Our algorithms for sliding windows will work for both sequence-based and time-based cases.  The only difference is that the definitions of the expiration of a point are different in the two cases.

\paragraph{Our Contributions}
This paper makes the following theoretical contributions.
\begin{enumerate}
\item We propose a robust $\ell_0$-sampling algorithm for well-separated datasets in the streaming model in constant dimensional Euclidean spaces; the algorithm uses $O(\log m)$ words of space ($m$ is the length of the stream) and $O(\log m)$ processing time per point, and successes with probability $(1 - 1/m)$ during the whole streaming process.  This result matches the one in the corresponding noiseless data setting.   See Section~\ref{sec:IW}

\item We next design an algorithm for sliding windows under the same setting. The algorithm works for both sequence-based and time-based sliding windows, using $O(\log n \log w)$ words of space and $O(\log n\log w)$ processing time per point with success probability $(1 - 1/m)$ during the whole streaming process.  We comment that the sliding window algorithm is much more complicated than the one for the infinite window, and is our main technical contribution.
See Section~\ref{sec:SW}.

\item For general datasets, we manage to show that the proposed $\ell_0$-sampling algorithms for well-separated datasets still produce almost uniform samples on general datasets. More precisely, it achieves the guarantee (\ref{eq:def-2}). See Section~\ref{sec:general}. 

\item We further show that our algorithms can also handle datasets in high dimensional Euclidean spaces given sufficiently large separation ratios. See Section~\ref{sec:highD}.

\item Finally, we show that our $\ell_0$-sampling algorithms can be used to efficiently estimate $F_0$ in both the standard streaming model and the sliding window models. See Section~\ref{sec:F0}.
\end{enumerate}
We have also implemented and tested our $\ell_0$-sampling algorithm for the infinite window case, and verified its effectiveness on various datasets. See Section~\ref{sec:exp}.

\paragraph{Related Work}
We now briefly survey related works on distinct sampling, and previous work dealing with datasets with near-duplicates.

The problem of $\ell_0$-sampling is among the most well studied problems in the data stream literature. It was first investigated in \cite{GT01,CMR05,FIS08}, and the current best result is due to Jowhari et al.~\cite{JST11}.  We refer readers to \cite{CF14} for an overview of a number of $\ell_0$-samplers under a unified framework. Besides being used in various statistical estimations~\cite{CMR05}, $\ell_0$-sampling finds applications in dynamic geometric problems (e.g., $\eps$-approximation, minimum spanning tree~\cite{FIS08}), and dynamic graph streaming algorithms (e.g., connectivity~\cite{AGM12a}, graph sparsifiers~\cite{AGM12b,AGM13}, vertex cover~\cite{CCHM15,CCEHMMV16}  maximum matching \cite{AGM12a,Konrad15,AKLY16,CCEHMMV16}, etc; see \cite{McGregor14} for a survey).  
However, all the algorithms for $\ell_0$-sampling proposed in the literature only work for noiseless streaming datasets.

$\ell_0$-sampling in the sliding windows on noiseless datasets can be done by running the algorithm in \cite{BDM02} with the rank of each item being generated by a random hash function.  As before, this approach cannot work for datasets with near-duplicates simply because the hash values assigned to near-duplicates will be different.  

$\ell_0$-sampling has also been studied in the distributed streaming setting~\cite{CT15} where there are multiple streams and we want to maintain a distinct sample over the union of the streams. The sampling algorithm in \cite{CT15} is essentially an extension of the random sampling algorithms in \cite{CMYZ12,TW11} by using a hash function to generate random ranks for items, and is thus again unsuitable for datasets with near-duplicates.

The list of works for $F_0$ estimation is even longer (e.g., \cite{FM85,BJKST02,DF03,Ganguly07,FFG+08,KNW10}; just mention a few).  Estimating $F_0$ in the sliding window model was studied in \cite{ZZCL10}. Again, all these works target noiseless data.

The general problem of processing noisy data streams {\em without} a comprehensive data cleaning step was only studied fairly recently \cite{CZ16} for the $F_0$ problem.  A number of statistical problems ($F_0$, $\ell_0$-sampling, heavy hitters, etc.) were studied in the distributed model under the same noisy data model \cite{Zhang15}.  Unfortunately the multi-round algorithms designed in the distributed model cannot be used in the data stream model because on data streams we can only scan the whole dataset once without looking back.

This line of research is closely related  to {\em entity resolution} (also called {\em data deduplication}, {\em record linkage}, etc.); see, e.g., \cite{KSS06,EIV07,HSW07,DN09}. However, all these works target finding and merging all the near-duplicates, and thus cannot be applied to the data stream model where we only have a small memory space and cannot store all the items.

\paragraph{Techniques Overview}
The high level idea of our algorithm for the infinite window is very simple. Suppose we can modify the stream by only keeping one representative point (e.g., the first point according to the order of the data stream) of each group, then we can just perform a uniform random sampling on the representative points, which can be done for example by the following folklore algorithm: We assign each point with a random rank in $(0, 1)$, and maintain the point with the minimum rank as the sample during the streaming process.  Now the question becomes: 
\begin{quote}
Can we identify (not necessarily store) the first point of each group space-efficiently?
\end{quote}

Unfortunately, we will need to use $\Omega(n)$ space ($n$ is the number of groups) to identify the first point of each group for a noisy streaming dataset, since we have to store at least $1$ bit to ``record'' the first point of each group to avoid selecting other later-coming points of the same group.   One way to deal with this challenge is to subsample a set of groups {\em in advance}, and then only focus on the first points of this set of groups.  Two issues remain to be dealt with: 
\begin{enumerate}
\item How to sample a set of groups in advance? 
\item How to determine the sample rate?  
\end{enumerate}
Note that before we have seen all points in the group, the group itself is {\em not} well-defined, and thus it is difficult to assign an ID to a group at the beginning and perform the subsampling. Moreover, the number of groups will keep increasing as we see more points,  we therefore have to decrease the sample rate along the way to keep the small space usage. 

For the first question, the idea is to post a random grid of side length $\Theta(\alpha)$ ($\alpha$ is the group distance threshold) upon the point set, and then sample cells of the grid instead of groups using a hash function.  We then say a group 
\begin{enumerate}
\item
$G$ is {\em sampled} if and only if $G$'s {\em first} point falls into a sampled cell, 

\item $G$ is {\em rejected} if $G$ has a point in a sampled cell, however the $G$'s first point is not in a sampled cell.  

\item $G$ is {\em ignored} if $G$ has no point in a sampled cell.
\end{enumerate}
We note that the second item is critical since we want to judge a group  only by its first point; even there is another point in the group that is sampled, if it is not the first point of the group, then we will still consider the group as rejected.  On the other hand, we do not need to worry about those ignored groups since they are not considered at the very beginning.

To guarantee that our decision is consistent on each group we have to keep some {\em neighborhood} information on each rejected group as well to avoid ``double-counting'', which seems to be space-expensive at the first glance.  Fortunately, for constant dimensional Euclidean space, we can show that if grid cells are randomly sampled, then the number of non-sampled groups is within a constant factor of that of sampled groups.  
We thus can control the space usage of the algorithm by dynamically decreasing the sample rate for grid cells.  More precisely, we try to maintain a sample rate as low as possible while guarantee that there is at least one group that is sampled. This answers the second question.

The situation in the sliding window case becomes complicated because points will expire, and consequently we cannot keep decreasing the grid cell sample rate.  In fact, we have to increase the cell sample rate when there are not enough groups being sampled.  However, if we increase the cell sample rate in the middle of the process, then the neighborhood information of those previously ignored groups has already got lost.  To handle this dilemma we choose to maintain a hierarchical sampling structure. We choose to describe the high level ideas as well as the actual algorithm in Section~\ref{sec:alg-SW} after the some basic algorithms and concepts have been introduced.

For general datasets, we show that our algorithms for well-separated datasets can still return an almost uniform random distinct sample. We first relate our robust $\ell_0$-sampling algorithm to a greedy partition process, and show that our algorithm will return a random group among the groups generated by that greedy partition.  We then compare that particular greedy partition with the minimum cardinality partition, and show that the number of groups produced by the two partitions are within a constant factor of each other.   
\smallskip

{\em Comparison with \cite{CZ16}.}  We note that although this work follows the noisy data model of that in \cite{CZ16} and the roadmap of this paper is similar to that of \cite{CZ16} (which we think is the natural way for the presentation), the contents of this paper, namely, the ideas, proposed algorithms, and analysis techniques, are all very different from that in \cite{CZ16}. After all, the $\ell_0$-sampling problem studied in this paper is different from the $F_0$ estimation studied in \cite{CZ16}.  We note, however, that there are natural connections between distinct elements and distinct sampling, and thus would like to mention a few points.
\begin{enumerate}
\item  In the infinite window case, we can easily use our robust $\ell_0$-sampling algorithm to get an algorithm for $(1+\eps)$-approximating robust $F_0$ using the same amount of space as that in \cite{CZ16} (see Section~\ref{sec:F0}).  We note that in the noiseless data setting, the problem of $\ell_0$-sampling and $F_0$ estimation can be reduced to each other by easy reductions.  However, it is not clear how to straightforwardly use $F_0$ estimation to perform $\ell_0$-sampling in the noisy data setting using the same amount of space as we have achieved.  We believe that since there is no magic hash function, similar procedure like finding the representative point of each group is necessary in any $\ell_0$-sampling algorithm in the noisy data setting.

\item Our sliding window $\ell_0$-sampling algorithm can also be used to obtain a sliding window algorithm for $(1+\eps)$-approximating $F_0$ (also see Section~\ref{sec:F0}). However, it is not clear how to extend the $F_0$ algorithm in \cite{CZ16} to the sliding window case, which was not studied in \cite{CZ16}.

\item In order to deal with general datasets, in \cite{CZ16} the authors introduced a concept called $F_0$-ambiguity and used it as a parameter in the analysis.  Intuitively, $F_0$-ambiguity measures the least fraction of points that we need to remove in order to make the dataset to be well-separated.  This definition works for problems whose answer is a single number, which does {\em not} depend on the actual group partition.  However, different group partitions do affect the result of $\ell_0$-sampling, even that all those partitions have the minimum cardinality.
In Section~\ref{sec:general} we show that by introducing a relaxed version of random sampling we can bypass the issue of data ambiguity.
\end{enumerate}

\paragraph{Preliminaries}
In Table~\ref{tab:notation} we summarize the main notations used in this paper.  We use $[n]$ to denote $\{1, 2, \ldots, n\}$.

We say $x$ is $(1+\eps)$-approximation of $y$ if $x \in [(1-\eps)y, (1+\eps)y]$.

\begin{table}[t]
 \centering
 \begin{tabular}{| l | l |}
	\hline
        Notation       &  Definition\\
   \hline
 	$S$  & stream of points \\
    \hline
   $m$   & length of the stream      \\
    \hline
   $w$   & length of the sliding window      \\
	\hline
   $n = F_0(S)$ & number of groups  \\
   \hline
   $\mathcal{G} / G$ & set of groups / a group\\
   \hline
   $G(p)$ & group containing point $p$\\
   \hline
    $\alpha$ & threshold of group diameter\\
   \hline
   $\mathbb{G} / C$ & grid / a grid cell  \\
   \hline
   $\text{\cell}(p)$ & cell containing point $p$ \\
   \hline
    $\text{\adj}(p)$ & set of cells adjacent to \cell$(p)$\\  
   \hline
   $\mathtt{Ball}(p, \alpha)$ & $\{q\ |\ d(p, q) \le \alpha\}$ \\
   \hline
    $\eps$ & approximation ratio for $F_0$ \\  
   \hline
 \end{tabular}
\caption{Notations}
 \label{tab:notation}
\end{table}

We need the following versions of the Chernoff bound.

\begin{lemma}[Standard Chernoff Bound]
\label{lem:Chernoff}
Let $X_1, \ldots, X_n$ be independent Bernoulli random variables such that $\Pr[X_i = 1] = p_i$. Let $X = \sum_{i \in [n]} X_i$. Let $\mu = \bE[X]$. It holds that $\Pr[X \ge (1+\delta)\mu] \le e^{-\delta^2\mu/3}$ and $\Pr[X \le (1-\delta)\mu] \le e^{-\delta^2\mu/2}$ for any $\delta \in (0,1)$.
\end{lemma}

\begin{lemma}[Variant of Chernoff bound]
 \label{lem:Hoeffding}
Let $Y_1, \ldots, Y_n$ be $n$ independent random variables
such that $Y_i \in [0,T]$ for some $T > 0$. Let $\mu = \bE[\sum_i
  Y_i]$. Then for any $a > 0$, we have
$$\Pr\left[\sum_{i\in[n]} Y_i > a\right] \le e^{-(a-2\mu)/T}.$$
\end{lemma}

\section{Well-Separated Datasets in Constant Dimensions}
\label{sec:2D}

We start with the discussion of $\ell_0$-sampling on well-separated datasets in constant dimensional Euclidean space.

\subsection{Infinite Window}
\label{sec:IW}

We first consider the infinite window case. We present our algorithm for $2$-dimensional Euclidean space, but it can be trivially extended to $O(1)$-dimensions by appropriately changing the constant parameters.


Let $\G = \{G_1, \ldots, G_n\}$ be the natural group partition of the well-separated stream of points $S$.  
We post a random grid $\Grid$ with side length $\frac{\alpha}{2}$ on $\mathbb{R}^2$, and call each grid square a {\em cell}.  For a point $p$, define $\cell(p)$ to be the cell $C \in \Grid$ containing $p$. Let 
$$\adj(p) = \{C \in \Grid\ |\ d(p, C) \le \alpha\},$$ 
where $d(p, C)$ is defined to be the minimum distance between $p$ and a point in $C$. We say a group $G$ {\em intersects} a cell $C$ if $G \cap C \neq \emptyset$.

Assuming that all points have $x$ and $y$ coordinates in the range $[0, M]$ for a large enough value $M$. Let $\Delta = \frac{2M}{\alpha}+1$. We assign the cell on the $i$-th row and the $j$-th column of the grid $\mathbb{G} \cap [0,M] \times [0,M]$ a numerical identification (ID) $((i-1) \cdot \Delta + j)$.  For convenience we will use ``cell'' and its ID interchangeably throughout the paper when there is no confusion.

For ease of presentation, we will assume that we can use fully random hash functions for free.  In fact, by Chernoff-Hoeffding bounds for limited independence \cite{SSS95,DP09}, all our analysis still holds when we adopt $\Theta(\log m)$-wise independent hash functions, using which will not affect the asymptotic space and time costs of our algorithms.  

Let $h: [\Delta^2] \to  \{0, 1, \ldots, 2^{\lceil 2 \log\Delta \rceil}-1\}$ be a fully random hash function,
and define $h_R$ for a given parameter $R = 2^k\ (k \in \mathbb{N})$ to be $h_R(x) = h(x) \bmod R$.  We will use  $h_R$ to perform sampling.  In particular, given a set of IDs $Y = \{y_1, \ldots, y_t\}$, we call $\{y \in Y\ |\ h_R(y) = 0\}$ the set of sampled IDs of $Y$ by $h_R$.  We also call $1/R$ the {\em sample rate} of $h_R$. 

\smallskip

As discussed in the techniques overview in the introduction,  our main idea is to sample {\em cells} instead of groups in advance using a hash function. 

\begin{definition}[sampled cell]
A cell $C$ is {\em sampled} by $h_R$ if and only if $h_R(C) = 0$.
\end{definition}

By our choices of the grid cell size and the hash function we have:

\begin{fact} (a) Each cell will intersect at most one group, and each group will intersect at most $O(1)$ cells. 

(b) For any set of points $P = \{p_1, \ldots, p_t\}$, 
$$\{p \in P\ |\ h_{2R}(\cell(p)) = 0\} \subseteq \{p \in P\ |\ h_{R}(\cell(p)) = 0\}.$$
\end{fact}

In the infinite window case (this section) we choose the {\em representative} point of each group to be the {\em first} point of the group.  We note that the representative points are fully determined by the input stream, and are independent of the hash function.  We will define the representative point slightly differently in the sliding window case (next section).

We define a few sets which we will maintain in our algorithms. 
\begin{definition}
\label{def:set}
Let $\Sf$ be the set of representative points of all groups.
Define the {\em accept set} to be
$$\Sa  = \{p \in \Sf ~|~ h_R(\cell(p)) = 0\},$$
and the {\em reject set} to be
$$\Sr  = \{p \in \Sf \backslash \Sa ~|~ \exists C \in \adj(p)~ \text{s.t. } h_R(C) = 0\}.$$
\end{definition}

For convenience we also introduce the following concepts.
\begin{definition}[sampled, rejected, candidate group]
\label{def:group}
We say a group $G$ a {\em sampled} group if $G \cap \Sa \neq \emptyset$, a {\em rejected} group if $G \cap \Sr \neq \emptyset$, and a {\em candidate} group if $G \cap (\Sa \cup \Sr) \neq \emptyset$.
\end{definition}

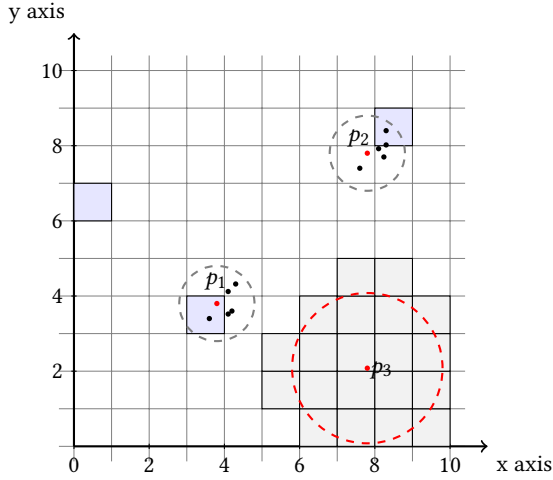
\begin{figure}
  \centering
  \begin{tikzpicture}
  \draw[step=0.5cm,gray,very thin] (-0.2, -0.2) grid (5.2,5.2);
  \draw[thick,->] (0,0) -- (5.5,0) node[anchor=north west] {x axis};
  \draw[thick,->] (0,0) -- (0,5.5) node[anchor=south east] {y axis};
  \foreach \x in {0,2,4,6, 8, 10}
  \draw (\x * 0.5 cm,1pt) -- (\x * 0.5 cm,-1pt) node[anchor=north] {$\x$};
  \foreach \y in {2,4,6, 8, 10}
  \draw (1pt,\y * 0.5 cm) -- (-1pt,\y * 0.5 cm) node[anchor=east] {$\y$};

  \filldraw[fill=gray!10!white, draw=black] (2.5,0.5) rectangle (3,1);
  \filldraw[fill=gray!10!white, draw=black] (2.5,1) rectangle (3,1.5);

  \filldraw[fill=gray!10!white, draw=black] (3,0) rectangle (3.5,0.5);
  \filldraw[fill=gray!10!white, draw=black] (3,0.5) rectangle (3.5,1);
  \filldraw[fill=gray!10!white, draw=black] (3,1) rectangle (3.5,1.5);
  \filldraw[fill=gray!10!white, draw=black] (3,1.5) rectangle (3.5,2);

  \filldraw[fill=gray!10!white, draw=black] (3.5,0) rectangle (4,0.5);
  \filldraw[fill=gray!10!white, draw=black] (3.5,0.5) rectangle (4,1);
  \filldraw[fill=gray!10!white, draw=black] (3.5,1) rectangle (4,1.5);
  \filldraw[fill=gray!10!white, draw=black] (3.5,1.5) rectangle (4,2);
  \filldraw[fill=gray!10!white, draw=black] (3.5,2) rectangle (4,2.5);

  \filldraw[fill=gray!10!white, draw=black] (4,0) rectangle (4.5,0.5);
  \filldraw[fill=gray!10!white, draw=black] (4,0.5) rectangle (4.5,1);
  \filldraw[fill=gray!10!white, draw=black] (4,1) rectangle (4.5,1.5);
  \filldraw[fill=gray!10!white, draw=black] (4,1.5) rectangle (4.5,2);
  \filldraw[fill=gray!10!white, draw=black] (4,2) rectangle (4.5,2.5);

  \filldraw[fill=gray!10!white, draw=black] (4.5,0) rectangle (5,0.5);
  \filldraw[fill=gray!10!white, draw=black] (4.5,0.5) rectangle (5,1);
  \filldraw[fill=gray!10!white, draw=black] (4.5,1) rectangle (5,1.5);
  \filldraw[fill=gray!10!white, draw=black] (4.5,1.5) rectangle (5,2);

  \draw (4.1, 1.02) node{$p_3$};
  \fill[color=red] (3.9, 1.04) circle(1pt);
  \draw[red,thick,dashed] (3.9,1.04) circle (1);

  \filldraw[fill=blue!10!white, draw=black] (1.5,1.5) rectangle (2,2);
  \draw (1.9, 2.2) node{$p_1$};
  \fill[color=red] (1.9, 1.9) circle(1pt);
  \draw[gray,thick,dashed] (1.9,1.9) circle (0.5);
  \fill (1.8, 1.7) circle(1pt);
  \fill (2.1, 1.8) circle(1pt);
  \fill (2.05, 1.76) circle(1pt);
  \fill (2.05, 2.06) circle(1pt);
  \fill (2.15, 2.16) circle(1pt);

  \filldraw[fill=blue!10!white, draw=black] (4,4) rectangle (4.5,4.5);
  \draw (3.8, 4.1) node{$p_2$};
  \fill[color=red] (3.9, 3.9) circle(1pt);
  \draw[gray,thick,dashed] (3.9,3.9) circle (0.5);
  \fill (3.8, 3.7) circle(1pt);
  \fill (4.12, 3.85) circle(1pt);
  \fill (4.05, 3.96) circle(1pt);
  \fill (4.15, 4.01) circle(1pt);
  \fill (4.15, 4.2) circle(1pt);

  \filldraw[fill=blue!10!white, draw=black] (0,3) rectangle (0.5,3.5);

\end{tikzpicture}

\caption{
  Each square is a cell;
  each light blue square is a sampled cell.  Each gray dash circle stands for a group.
  Red points ($p_1, p_2$ and $p_3$) are representative points;
  $p_1$ is in the accept set and $p_2$ is in the reject set.
  Gray cells form $\adj(p_3)$.  $\alpha = 2$ in this illustration.}
  \label{fig:cell}
\end{figure}

Figure~\ref{fig:cell} illustrates some of the concepts we have introduced so far.

Obviously, the set of sampled groups and the set of rejected groups are disjoint, and their union is the set of candidate groups.
Also note that $\Sa$ is the set of representative points of the sampled groups, and $\Sr$ is the set of representative points of rejected groups.  

We comment that it is important to keep the set $\Sr$, even that at the end we will only sample a point from $\Sa$. This is because otherwise we will have no information regarding the first points of those groups that may have points other than the first ones falling into a sampled cell, and consequently points in $S \backslash \Sf$ may also be included into $\Sa$, which will make the final sampling to be non-uniform among the groups.
One may wonder whether this additional storage will cost too much space.  Fortunately, since each group has diameter at most $\alpha$, we only need to monitor groups that are at a distance of at most $\alpha$ away from sampled cells, whose cardinality can be shown to be small. More precisely, for a group $G$, letting $p$ be its representative point, we monitor $G$ if and only if there exists a sampled cell $C$ such that $C \in \adj(p)$.  The set of representative points of such groups is precisely $\Sa \cup \Sr$.

\begin{algorithm}[t]
  \DontPrintSemicolon
  $R \gets 1$; $\Sa \gets \emptyset$; $\Sr \gets \emptyset$\;
  $\kappa_0$ is chosen to be a large enough constant \;
  \tcc{dataset is fed as a stream of points}
  \For{each arriving point $p$}{\label{line:a-1}
  	\tcc{if $p$ is not the first point of a candidate group, skip it}
    \If{$\exists u \in \Sa \cup \Sr$ s.t. $d(u, p) \le \alpha$ \label{line:a-2}}{
      continue\;
    }
    \tcc{if $p$ is the first point of a candidate group}
    \If{$h_R(\cell(p)) = 0$ \label{line:a-3}}{
      $\Sa \gets \Sa \cup \{p\}$\;     
    }\ElseIf{$\exists C \in \adj(p)$ s.t. $h_R(\cell(C)) = 0$ \label{line:a-4}}{
      $\Sr \gets \Sr \cup \{p\}$\;
    }
    \If{$|\Sa| > \kappa_0 \log m$\label{line:a-5} }{
      $R \gets 2 R$\;
      update $\Sa$ and $\Sr$ according to the updated hash function $h_R$\;
    }
  }
\tcc{at the time of query:}
\Return a random point in $\Sa$\;
\caption{{\sc Robust} $\ell_0$-{\sc Sampling-IW}}
\label{alg:IW}
\end{algorithm}

Our algorithm for $\ell_0$-sampling in the infinite window case is presented in Algorithm~\ref{alg:IW}.  We control the sample rate by doubling the range $R$ of the hash function when the number of points of $\Sa$ exceeds a threshold $\Theta(\log m)$ (Line \ref{line:a-5} of Algorithm~\ref{alg:IW}). We will also update $\Sa$ and $\Sr$ accordingly to maintain Definition~\ref{def:set}.  

When a new point $p$ comes, if $\cell(p)$ is sampled and $p$ is the first point in $G(p)$ (Line \ref{line:a-3}), we add $p$ to $\Sa$; that is, we make $p$ as the representative point of the sampled group $G(p)$.  Otherwise if $\cell(p)$ is not sampled but there is at least one sampled cell in $\adj(p)$, and $p$ is the first point in $G(p)$ (Line \ref{line:a-4}), then we add $p$ to $\Sr$; that is, we make $p$ as the representative point of the rejected group $G(p)$.  

On the other hand, if there is at least one sampled cell in $\adj(p)$ (i.e., $G(p)$ is a candidate group) and $p$ is {\em not} the first point (Line \ref{line:a-2}), then we simply discard $p$.  Note that we can test this since we have already stored the representation points of all candidate groups.
In the remaining case in which $G(p)$ is not a candidate group, we also discard $p$. 

At the time of query, we return a random point in $\Sa$.

\paragraph{Correctness and Complexity}
We show the following theorem.
\begin{theorem}
\label{thm:IW}
In constant dimensional Euclidean space for a well-separated dataset, there exists a streaming algorithm (Algorithm~\ref{alg:IW}) that with probability $1 - 1/m$, at any time step, it outputs a robust $\ell_0$-sample. The algorithm uses $O(\log m)$ words of space and $O(\log m)$  processing time per point.
\end{theorem}

The correctness of the algorithm is straightforward. First, $\Sa$ is a random subset of $\Sf$ since each point $p \in \Sf$ is included in $\Sa$ if and only if  $h_R(\cell(p)) = 0$. Second, the outputted point is a random point in $\Sa$.  The only thing left to be shown is that we have $\abs{\Sa} > 0$ at any time step.

\begin{lemma}
\label{lem:0}
With probability $1 - 1/m$, we have $\abs{\Sa} > 0$ throughout the execution of the algorithm.
\end{lemma} 

\begin{proof}
At the first time step of the streaming process, $p_1$ is added into $\Sa$ with certainty since $R$ is initialized to be $1$.  Then $\Sa$ keeps growing.  At the moment when $\abs{\Sa} > \kappa_0 \log m$, $R$ is doubled so that each point in $\Sa$ is resampled with probability $\frac{1}{2}$.  After the resampling, 
\begin{equation}
\label{eq:b-0}
\Pr[\abs{\Sa} = 0] \le \left(\frac{1}{2}\right)^{\kappa_0 \log m} \le \frac{1}{m^2}.
\end{equation}
By a union bound over at most $m$ resample steps, we conclude that with probability $1 - 1/m$, $\abs{\Sa} > 0$ throughout the execution of the algorithm.  
\end{proof}

We next analyze the space and time complexities of Algorithm~\ref{alg:IW}.

\begin{lemma}
\label{lem:Sr}
With probability $(1 - 1/m)$ we have $\Sr = O(\log m)$ throughout the execution of the algorithm.
\end{lemma}

\begin{proof}
Consider a fixed time step.  Let $S = \Sa \cup \Sr$.  For a fixed $p \in \Sf$, since $\abs{\adj(p)} \le 25$ (we are in the $2$-dimensional Euclidean space), and each cell is sampled randomly, we have
\begin{equation}
\label{eq:b-1}
\Pr[p \in \Sr] \le \frac{24}{25} \cdot \Pr[p \in S].
\end{equation}
We only need to prove the lemma for the case $\Pr[p \in \Sr] = \frac{24}{25} \cdot \Pr[p \in S]$; the case $\Pr[p \in \Sr] < \frac{24}{25} \cdot \Pr[p \in S]$ follows directly since $p$ is less likely to be included in $\Sr$.

For each $p \in S$, define $X_p$ to be a random variable such that $X_p = 1$ if $p \in \Sr$, and $X_p = 0$ otherwise.  Let $X = \sum_{p \in \Sr} X_p$.  We have $X = \abs{\Sr}$ and $\bE[X] = \frac{24}{25} \abs{S}$.  By a Chernoff bound (Lemma~\ref{lem:Chernoff}), we have
\begin{equation}
\label{eq:b-2}
\Pr\left[X - \bE[X] > 0.01 \bE[X]\right] \le e^{- \frac{0.01^2 \cdot \bE[X]}{3}}.
\end{equation}
If $\abs{S} \le 80000 \log m$ then we immediately have $\abs{\Sr} \le \abs{S} = O(\log m)$. Otherwise by (\ref{eq:b-2}) we have 
$$\Pr[X > 1.01 \bE[X]] < 1/m^2.$$ We thus have
\begin{eqnarray*}
1/m^2 &>& \Pr[X > 1.01 \bE[X]] \\
&=& \Pr[X > 1.01 \cdot \frac{24}{25} \abs{S}] \\
&=& \Pr[X > 0.9696(X + \abs{\Sa})] \\
&=& \Pr[0.0304 X > 0.9696 \abs{\Sa}].
\end{eqnarray*}
According to Algorithm~\ref{alg:IW} it always holds that $\abs{\Sa} = O(\log m)$. Therefore $\abs{\Sr} = X = O(\log m)$ with probability at least $1 - 1/m^2$.  The lemma follows by a union bound over $m$ time steps of the streaming process.
\end{proof}

By Lemma~\ref{lem:Sr} the space used by the algorithm can be bounded by $O(\abs{\Sa} + \abs{\Sr}) = O(\log m)$ words.  The processing time per point is also bounded by $O(\abs{\Sa} + \abs{\Sr})$.

\subsection{Sliding Windows}
\label{sec:SW}

We now consider the sliding window case. Let $w$ be the window size.  We first present an algorithm that maintains a set of sampled points in $\Sa$ with a fixed sample rate $1/R$; it will be used as a subroutine in our final sliding window algorithm (Section~\ref{sec:alg-SW}).  

\subsubsection{A Sliding Window Algorithm with Fixed Sample Rate}
\label{sec:SW-R}

We describe the algorithm in Algorithm~\ref{alg:SW-R}, and explain it in words below.  


\begin{algorithm}[t]
  \DontPrintSemicolon
  \For{each expired point $p$}{
    \If{$\exists (u, p) \in A$}{
      delete $(u, p)$ from $A$, delete $u$ from $\Sa \cup \Sr$ \label{line:b-3}
    } 		

  }
  
  \For{each arriving point $p$}{
    \tcc{if there already exists a point of the same group in $\Sa \cup \Sr$}
    \If{$\exists u \in \Sa \cup \Sr$ ~s.t. $d(u, p) \leq \alpha$}{
       		$A \gets (u, p) \cup A \backslash (u, \cdot)$ \label{line:b-4}
       }
    \tcc{otherwise we set $p$ as a representative of its group}
    \ElseIf{$h_R(\cell(p)) = 0$ \label{line:b-7}}{
      $\Sa \gets \Sa \cup \{p\}, A \gets A \cup (p, p)$ \label{line:b-5}
    }
    \ElseIf{$\exists C \in \adj(p)$ s.t. $h_R(C) = 0$ \label{line:b-8}}{
      $\Sr \gets \Sr \cup \{p\}, A \gets A \cup (p, p)$ \label{line:b-6}\;
    }
  }
\caption{{\sc SW with Fixed Sample Rate $1/R$}}
\label{alg:SW-R}
\end{algorithm}

Besides maintaining the accept set and the reject set as that in the infinite window case, Algorithm~\ref{alg:SW-R} also maintains a set $A$ consisting of key-value pairs $(u, p)$, where $u$ is the {\em representative} point of a candidate group ($u$ can be a point outside the sliding window as long as the group has at least one point inside the sliding window), and $p$ is the {\em latest} point of the same group (thus $p$ must be in the sliding window).  Define $A(\Sa) = \{p\ |\ \exists u  \in \Sa  \text{ s.t. } (u, p) \in A\}$. 


For each sliding window, we guarantee that each candidate group $G$ has exactly one representative point.  This is achieved by the following process: for each candidate group $G$, if there is no maintained representative point, then we take the first point $u$ as the representative point (Line~\ref{line:b-5} and Line~\ref{line:b-6}).  When the last point $p$ of the group expires, we delete the maintained representative point $u$ from $\Sa\cup \Sr$, and delete $(u, p)$ from $A$ (Line~\ref{line:b-3}). 

For a new arriving point $p$, if there already exists a point $u \in \Sa \cup \Sr$ in the same group $G$, then we simply update the last point in the pair $(u, \cdot)$ we maintained for $G$ (Line~\ref{line:b-4}).  Otherwise $p$ is the first point of $G(p)$ in the sliding window. If $G(p)$ is a sampled group, then we add $p$ to $\Sa$ and $(p, p)$ to $A$ (Line~\ref{line:b-5}); else if $G(p)$ is a rejected group, then we add $p$ to 
$\Sr$ and $(p, p)$ to $A$ (Line~\ref{line:b-6}).

The following observation is a direct consequence of Algorithm~\ref{alg:SW-R}. It follows from the discussion above and the testing at Line~\ref{line:b-7} of Algorithm~\ref{alg:SW-R}.

\begin{figure}
  \centering
  \includegraphics[width=0.4\textwidth]{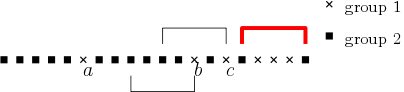}
  \caption{Representative points in sliding windows. There are two different groups, and the red window is the current sliding window (of size $w = 5$). Point $c$ is {\em not} the representative point of Group $1$ because the window right before it (inclusive) contains point $b$ which is also in Group $1$. Point $b$ is the representative point because it is the \emph{latest} point such that there is no other point of Group $1$ in the window right before $b$.} 
  \label{fig:rep}
\end{figure}

\begin{observation}
\label{lem:SW-R}
In Algorithm~\ref{alg:SW-R}, at any time for the current sliding window, we have
\begin{enumerate}
\item Each group has exactly one representative point, which is fully determined by the stream and is {\em independent} of the hash function. More precisely, a point $p$ becomes the representative point of group $G$ in the current window if $p$ is the latest point in $G$ such that the window right before $p$ (inclusive) has no point in $G$.  See Figure~\ref{fig:rep} for an illustration.  

\item The representative point of each group in the current window is included in $\Sa$ with probability $1/R$.
\end{enumerate}
\end{observation}

\subsubsection{A Space-Efficient Algorithm for Sliding Windows}
\label{sec:alg-SW}

We now present our space-efficient sliding window algorithm. Note that the algorithm presented in Section~\ref{sec:SW-R}, though being able to produce a random sample in the sliding window setting, does not have a good space usage guarantee; it may use space up to $w/R$ where $w$ is the window size.  

The sliding window algorithm presented in this section works simultaneously for both sequence-based sliding window and time-based sliding window.

\paragraph{High Level Ideas}
As mentioned, the main challenge of the sliding window algorithm design is that points will expire, and thus we cannot keep decreasing the sample rate.  On the contrary, if at some point there are too few non-expired sampled points, then we have to increase the sample rate to guarantee that there is at least one point in the sliding window belonging to $\Sa$.  However, if we increase the sample rate in the middle of the streaming process, then the neighborhood information of a newly sampled group may already get lost.  In other words, we cannot properly maintain $\Sr$ when the sample rate increases. 

To resolve this issue we have the ``prepare'' such a decrease of $\abs{\Sa}$ in advance.  To this end, we maintain a hierarchical set of instances of Algorithm~\ref{alg:SW-R}, with sample rates $1/R$ being $1, 1/2, 1/4, \ldots$ respectively.  We thus can guarantee that in the lowest level (the one with sample rate $1$) we must have at least one sampled point.  

Of course, to achieve a good space usage we cannot endlessly insert points to all the Algorithm~\ref{alg:SW-R} instances.  We instead  make sure that each level $\ell$ stores at most $\abs{\Sa_\ell \cup \Sr_\ell} = O(\log m)$ points, where $\Sa_\ell$ and $\Sr_\ell$ are the accept set and reject set respectively in the run of an Algorithm~\ref{alg:SW-R} instance at level $\ell$.  We achieve this by maintaining a {\em dynamic} partition of the sliding window. In the $\ell$-th subwindow we run an instance of Algorithm~\ref{alg:SW-R} with sample rate $1/2^\ell$.  For each incoming point,  we ``accept'' it at the highest level $\ell$ in which the point falls into $\Sa_\ell$, and then delete {\em all} points in the accept and reject sets in all the lower levels.  Whenever the number of points in $\Sa_\ell$ at some level $\ell$ exceeds the threshold $c \log m$ for some constant $c$, we ``promote'' most of its points to level $\ell+1$.  The process may cascade to the top level. 

At the time of query we properly resample the points maintained at each $\Sa_\ell\ (\ell = 0, 1, \ldots)$ to unify their sample probabilities, and then merge them to $\Sa$.  In order to guarantee that if the sliding window is not empty then we always have at least one sampled point in $\Sa$, during the algorithm (in particular, the promotion procedure) we make sure that the last point of each level $\ell$ is always in the accept set $\Sa_\ell$.

\begin{remark}
The hierarchical set of windows reminisces the {\em exponential histogram} technique by Datar et al.~\cite{DGIM02} for {\em basic counting} in the sliding window model. However, by a careful look one will notice that our algorithm is very different from exponential histogram, and is (naturally) more complicated since we need to deal with both {\em distinct elements} and {\em near-duplicates}. For example, the exponential histogram algorithm in \cite{DGIM02} partitions the sliding window {\em deterministically} to subwindows of size $1, 2, 4, \ldots$.  Suppose we are only interesting in the representative point of each group,  we basically need to delete all the other points in each group in the sliding window, which will change the sizes of the subwindows.  Handling near duplicates adds another layer of difficulty to the algorithm design; we handle this by employing Algorithm~\ref{alg:SW-R} (which is a variant of the algorithm for the infinite window in Section~\ref{sec:IW}) at each of the subwindows with different sample rates.  The interplay between these components make the overall algorithm involved.
\end{remark}

\paragraph{The Algorithm}
We present our sliding window algorithm in Algorithm~\ref{alg:SW} using Algorithm~\ref{alg:split} and Algorithm~\ref{alg:merge} as subroutines.

\begin{algorithm}[t]
\DontPrintSemicolon
$R_\ell \gets 2^\ell$ for all $\ell = 0, 1, \ldots, L$.

\For{$\ell \gets 0$ \KwTo $L$}{
\tcc{create an algorithm instance according to Algorithm~\ref{alg:SW-R} with fixed sample rate $1/R_\ell$}
	 $\alg_\ell \gets (\bot, \bot, \bot, R_\ell)$\;
}

\For{each arriving point $p$}{
	\For{$\ell \gets L$ {\bf downto} $0$ \label{line:c-1}}{
          $\alg_\ell(p)$ \tcc{feed $p$ to the instance $\alg_\ell$}
          
          \If{$\exists\ (u, p) \in A_\ell$}{
            \tcc{prune all subsequent levels}
			\For{$j \gets \ell - 1$ {\bf downto} $0$ \label{line:c-1-5}}{
				$\alg_j \gets (\bot, \bot, \bot, R_j)$  \label{line:c-2}
			}
			\If{$\abs{\Sa_\ell} > \kappa_0 \log m$  \label{line:c-3}}{
				$j \gets \ell$\;
				create a temporary instance $\alg$\;
				\While{$(|\Sa_j| > \kappa_0 \log m)$}{
					$(\alg, \alg_j) \gets$ {\sc Split}($\alg_j$)\;
					$\alg_{j+1} \gets$ {\sc Merge}($\alg, \alg_{j+1}$)	\;
					$j \gets j+1$  \label{line:c-4}\;
					\lIf{$j > L$}{\Return ``error''} \label{line:c-31}
				}
             }
		}
		 break
	}
} 

\tcc{at the time of query:}
$S \gets \emptyset$\; 
Let $c$ be the maximum index $\ell$ such that $\Sa_\ell \neq \bot$ \label{line:c-5} \; 
\For{$\ell \gets 1$ \KwTo $c$}{
 include each point in the set $\{p\ |\ \exists\ (\cdot, p) \in A_\ell \}$ to $S$ with probability $R_\ell/R_c$  \label{line:c-6}\; 
}
\Return a random point from $S$\;
\caption{Robust $\ell_0$-{\sc Sampling-SW}}
\label{alg:SW}
\end{algorithm}

\begin{algorithm}[t]
\DontPrintSemicolon

create instances $\alg_a = (\Sa_a, \Sr_a, A_a, R_a)$ and $\alg_b = (\Sa_b, \Sr_b, A_b, R_b)$\;
	$t = \max\{t' \ |\ (p_{t'} \in \Sa_\ell) \wedge (h_{R_{\ell+1}}(\cell(p_{t'})) = 0)\}$ \;
	$\Sa_a \gets \{p_k \in \Sa_\ell ~|~ (k \le t) \wedge (h_{R_{\ell+1}}(\cell(p_k)) = 0)\}$; $\Sr_a \gets \{p_k \in \Sr_\ell ~|~ (k \le t) \wedge (h_{R_{\ell+1}}(\cell(p_k)) = 0)\}$; $A_a \gets \{(u, \cdot)\in A_\ell ~|~ u \in \Sa_\ell\}$; $R_a \gets R_{\ell+1}$\;
	$\Sa_b \gets \{p_k \in \Sa_\ell ~|~ k > t\}$; $\Sr_b \gets \{p_k \in \Sr_\ell ~|~ k > t\}$; $A_b \gets \{(u, \cdot)\in A_\ell ~|~ u \in \Sa_\ell\}$; $R_b \gets R_\ell$\;

\Return ($\alg_a, \alg_b$)

\caption{{\sc Split($\alg_\ell$)}}
\label{alg:split}
\end{algorithm}

\begin{algorithm}[t]
\DontPrintSemicolon
create a temporary instance $\alg = (\Sa, \Sr, A, R)$\;
$\Sa \gets \Sa_a \cup \Sa_b$;
$\Sr \gets \Sr_a \cup \Sr_b$;
$A \gets A_a \cup A_b$;
$R \gets R_a$\;

\Return $\alg$

\caption{{\sc Merge($\alg_a, \alg_b$)}}
\label{alg:merge}
\end{algorithm}

We use $\alg$ to represent an instance of Algorithm~\ref{alg:SW-R}.  For convenience, we also use $\alg$ to represent all the data structures that are maintained during the run of Algorithm~\ref{alg:SW-R}, and write $\alg = (\Sa, \Sr, A, R)$, where $\Sa, \Sr$ are the accept and reject sets respectively, $A$ is the key-value pair store, and $R$ is the reciprocal of the sample rate.

Set $R_\ell = 2^\ell$ for $\ell = 0, 1, \ldots, L = \log w$.  In Algorithm~\ref{alg:SW} we create $L$ instances of Algorithm~\ref{alg:SW-R} with empty $\Sa_\ell, \Sr_\ell, A_\ell$ (denoted by `$\perp$'), and sample rates $1/R_\ell$ respectively.  We call the instance with $R_\ell = 2^\ell$ the level $\ell$ instance.

When a new point $p$ comes, we first find the highest level $\ell$ such that $p$ is sampled by $\alg_\ell$ (i.e., $p \in \Sa_\ell$), and then delete all the structures of $\alg_j\ (j < \ell)$, except keep their sample rates $1/R_j$ (Line~\ref{line:c-1} to Line~\ref{line:c-2}).  

If after including $p$, the size of $\Sa_\ell$ becomes more than $\kappa_0 \log m$, we have to do a series of updates to restore the invariant that the accept set of each level contains at most $\kappa_0 \log m$ points at any time step (Line~\ref{line:c-3} to Line~\ref{line:c-4}).  To do this, we first split the instance of $\alg_\ell$ into two instances (Algorithm~\ref{alg:split}).  Let point $p$ be the last point in $\Sa_\ell$ which is sampled by hash function $h_{R_{\ell+1}}$.  We promote all the points in $\Sa_\ell \cup \Sr_\ell$ arriving before (and include) $p$ to level $\ell+1$ by resampling them using $h_{R_{\ell+1}}$, which gives a new level $\ell+1$ instance $\alg$.  

We next try to merge $\alg$ with $\alg_{\ell+1}$ who have the same sample rate by merging their accept/reject sets and the sets of key-value pairs (Algorithm~\ref{alg:merge}).  The merge may result $\abs{\Sa_{\ell+1}} > \kappa_0 \log m$, in which case we have to perform the split and merge again.  These operations may propagate to the upper levels until we research a level $\ell$ in which we have $\abs{\Sa_\ell} \le \kappa_0 \log m$ after the merge.

At the time of query, we have to integrate the maintained samples in all $L$ levels.  Since at each level we sample points use different sample rates $1/R_\ell$, we have to resample each point in $\Sa_\ell$ with probability $R_\ell / R_c$ where $c$ is the largest level that has a non-empty accept set (Line~\ref{line:c-5} to Line~\ref{line:c-6}).

\paragraph{Correctness and Complexity}
In this section we prove the following theorem.  
\begin{theorem}
\label{thm:SW}
In constant dimensional Euclidean space for a well-separated dataset,
there exists a sliding window algorithm (Algorithm~\ref{alg:SW}) that with probability $1 - 1/m$, at any time step, it outputs a robust $\ell_0$-sample.  The algorithm uses $O(\log w \log m)$ words of space and $O(\log w \log m)$ amortized processing time per point.
\end{theorem}

First, it is easy to show the probability that Algorithm~\ref{alg:SW} outputs ``error'' is negligible.
\begin{lemma}
\label{lem:error}
With probability $1 - 1/m^2$, Algorithm~\ref{alg:SW} will not output ``error'' at Line~\ref{line:c-31} during the whole streaming process.
\end{lemma}

\begin{proof}
Fix a sliding window $W$.  Let $G_1, \ldots, G_k\ (k \le w)$ be the groups in $W$. The sample rate at level $L$ is $1/R_L = 1/2^L = 1/w$.  
Let $X_\ell$ be a random variable such that $X_\ell = 1$ if the $\ell$-th group is sampled, and $X_\ell = 0$ otherwise. Let $X = \sum_{\ell = 1}^k X_\ell$.  Thus $\bE[X] = k \cdot 1/R_L \le w \cdot 1/w = 1$. By a Chernoff bound (Lemma~\ref{lem:Hoeffding}) we have that with probability $1 - 1/m^3$, we have $X \le \kappa_0 \log m$ for a large enough constant $\kappa_0$.  The lemma then follows by a union bound over at most $m$ sampling steps.
\end{proof}

The following definition is useful for the analysis.
\begin{definition}[subwindows]
\label{def:subwindow}
For a fixed sliding window $W$, we define a subwindow $W_\ell$ for each instance $\alg_\ell\ (\ell = 0, 1, \ldots, L)$ as follows.  $W_L$ starts from the first point in the sliding window to the last point (denoted by $p_{t_L}$) in $A(\Sa_L)$. For $\ell = L - 1, \ldots, 1$, $W_\ell$ starts from $p_{t_{\ell + 1}+1}$ to the last point (denoted by $p_{t_\ell}$) in $A(\Sa_\ell)$.  $W_0$ starts from $p_{t_1+1}$ to the last point in the window $W$.  
\end{definition}

See Figure~\ref{fig:subwindow} for an illustration of subwindows.

\begin{figure}[t]
\centering
\includegraphics[height = 1.6in]{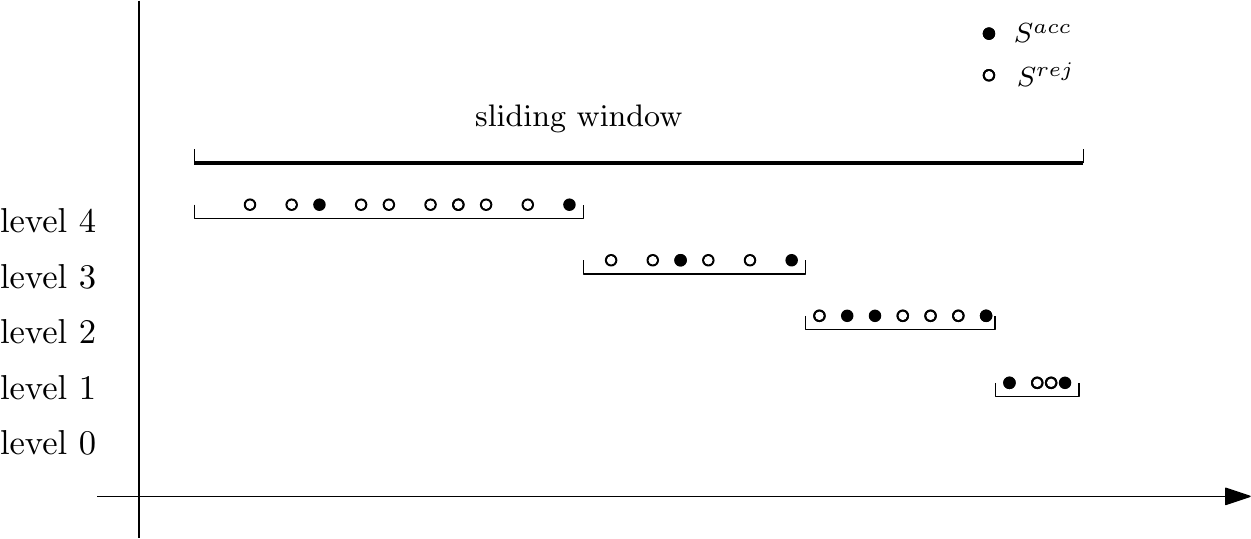}
\caption{An illustration of subwindows of a sliding window;
the subwindow at level $0$ is an empty set.}
\label{fig:subwindow}
\end{figure}

We note that a subwindow can be empty. We also note the following immediate facts by the definitions of subwindows.
\begin{fact}
$W_0 \cup W_1 \cup \ldots \cup W_L$ covers the whole window $W$.
\end{fact}

\begin{fact}
\label{fact:end-p}
Each subwindow $W_\ell\ (\ell = 1, \ldots, L)$ ends up with a point in $A(\Sa_\ell)$.
\end{fact}

For $\ell = 0, 1, \ldots, L$, let $\G_\ell$ be the set of groups whose last points lie in $W_\ell$, and let $\Sf_\ell$ be the set of their representative points.  From Algorithm~\ref{alg:SW}, \ref{alg:split} and \ref{alg:merge} it is easy to see that the following is maintained during the whole streaming process.

\begin{fact}
\label{fact:subwindow}
During the run of Algorithm~\ref{alg:SW}, at any time step, $\Sa_\ell$ is formed by sampling each point in $\Sf_\ell$ with probability ${1}/{R_\ell}$.
\end{fact}

The following lemma guarantees that at the time of query we can always output a sample.

\begin{lemma}
\label{lem:sample-exist}
During the run of Algorithm~\ref{alg:SW}, at any time step, if the sliding window contains at least one point, then when querying we can always return a sample, i.e., $\abs{S} \ge 1$.
\end{lemma}

\begin{proof}
The lemma follows from Fact~\ref{fact:end-p}, and the fact that $\alg_0$ includes every point in $\Sf_0$ ($R_0 = 1$).
\end{proof}

Now we are ready to prove the theorem.

\begin{proof}[(for Theorem~\ref{thm:SW})]
We have the following facts:
\begin{enumerate}
\item $\Sf_0, \Sf_1, \ldots, \Sf_L$ are set of representatives of {\em disjoint} sets of groups $\G_0, \G_1, \ldots, \G_L$.  And $\cup_{\ell = 0}^L \G_\ell$ is the set of all groups who have the last points inside the sliding window.

\item By Fact~\ref{fact:subwindow} each $\Sa_\ell$ is formed by sampling each point in $\Sf_\ell$ with probability $1/R_\ell$. 

\item Each point in $\Sf_\ell$ is included in $S$ with probability $R_\ell/R_c$ (Line~\ref{line:c-6} of Algorithm~\ref{alg:SW}).

\item By Lemma~\ref{lem:sample-exist}, $\abs{S} \ge 1$ if the sliding window is not empty.  \label{item-4}

\item The final sample returned is a random sample of $S$. \label{item-5}

\item By Lemma~\ref{lem:error}, with probability $(1 - 1/m)$ the algorithm will not output ``error''. \label{item-6}
\end{enumerate}
By the first three items we know that $S$ is a random sample of the last points of all groups within the sliding window, which, combined with Item \ref{item-4}, \ref{item-5} and \ref{item-6}, give the correctness of the theorem. 

We now analyze the space and time complexities.  The space usage at each level can be bounded by $O(\log m)$ words. This is due to the fact that $\abs{\Sa_j}$ is always no more $\kappa_0 \log m$, and consequently  $A_j$ has $O(\log m)$ key-value pairs. Using Lemma~\ref{lem:Sr} we have that with probability $1 - 1/m^2$, $\abs{\Sr_j} = O(\log m)$.\footnote{We can reduce the failure probability $1/m$ to $1/m^2$ by appropriately changing the constants in the proof.}  Thus by a union bound, with probability $(1 - O(\log w/m^2))$, the total space is bounded by $O(\log w \log m)$ words since we have $O(\log w)$ levels.

For the time cost, simply note that the time spent on each point at  each level during the whole streaming process can be bounded by $O(\log m)$, and thus the amortized processing time per item can be bounded by $O(\log w \log m)$. 
\end{proof}

\subsection{Discussions}
\label{sec:discussion}

We conclude the section with some discussions and easy extensions.

\paragraph{Sampling $k$ Points with/without Replacement}
Sampling $k$ groups {\em with} replacement can be trivially achieved by running $k$ instances of the algorithm for sampling one group (Algorithm~\ref{alg:IW} or Algorithm~\ref{alg:SW}) in parallel.  For sampling $k$ groups {\em without} replacement,  we can increase the threshold at Line~\ref{line:a-5} of Algorithm~\ref{alg:IW} to $\kappa_0  k \log m$, by which we can show using exactly the same analysis in Section~\ref{sec:IW} that with probability $1 - 1/m$ we have $\abs{\Sa} \ge k$.  Similarly, for the sliding window case we can increase the threshold at Line~\ref{line:c-3} of Algorithm~\ref{alg:SW} to $\kappa_0  k \log m$.

\paragraph{Random Point As Group Representative}  We can easily augment our algorithms such that instead of always returning the (fixed) representative point of a randomly sampled group, we can return a random point of the group. In other words, we want to return each point $p \in G$ with equal probability $\frac{1}{n \cdot \abs{G}}$. 

For the infinite window case we can simply plug-in the classical Reservoir sampling~\cite{Vitter85} in Algorithm~\ref{alg:IW}.  We can implement this as follows: For each group $G$ that has a point stored in $\Sa \cup \Sr$, we maintain an $e_G = (v, ct)$ pair where $ct$ is a counter counting the number of points of this group, and $v$ is the random representative point.  At the beginning (when the first point $u$ of group $G$ comes) we set $e_G = (u, 1)$. When a new point $p$ is inserted, if there exists $u \in \Sa$ such that $d(u, p) \le \alpha$ (i.e., $u$ and $p$ are in the same group), we increment the counter $ct$ for group $G(u)$, and reset $e_G = (u, p)$ with probability $\frac{1}{ct}$.  For the sliding window case, we can just replace Reservoir sampling with a random sampling algorithm for sliding windows (e.g., the one in \cite{BOZ09}).


\section{General Datasets}
\label{sec:general}

In this section we consider general datasets which may {\em not} be well-separated, and consequently there is {\em no} natural partition of groups. 
However, we show that Algorithm~\ref{alg:IW} still gives the following guarantee.

%

\begin{theorem}
\label{thm:ambiguity}
For a general dataset $S$ in constant dimensional Euclidean space, there exists a streaming algorithm (Algorithm~\ref{alg:IW}) that with probability $1 - 1/m$, at any time step, it outputs a point $q$ satisfying Equality~(\ref{eq:def-2}), that is, 
\begin{equation*} 
\textstyle
\forall p \in S, \Pr[q \in \ball(p, \alpha) \cap S] = \Theta(\frac{1}{F_0(S, \alpha)}),
\end{equation*}
where $\ball(p, \alpha)$ is the ball centered at $p$ with radius $\alpha$.  
\end{theorem}


Before proving the theorem, we first study group partitions generated by a greedy process.

\begin{definition}[Greedy Partition]
\label{def:greedy}
Given a dataset $S$, a greedy partition is generated by the following process: pick an arbitrary point $p \in S$, create a new group $G(p) \gets \ball(p, \alpha) \cap S$ and update $S \gets S \backslash G(p)$; repeat this process until $S = \emptyset$.
\end{definition}

\begin{lemma}
\label{lem:greedy}
Given a dataset $S$, let $\opt$ be the number of groups in the minimum cardinality partition of $S$, and $\sol$ be the number of groups in an {\em arbitrary} greedy partition.  We always have $\opt = \Theta(\sol)$.
\end{lemma}

\begin{proof}
We first show $\sol \le \opt$. Let $G(p_1), \ldots, G(p_{\sol})$ be the groups in the greedy partition according to the orders they were created, and let $H_1, \ldots, H_{\opt}$ be the minimum partition.  

We prove by induction.  First it is easy to see that $G(p_1)$ must cover the group containing $p_1$ in the minimum partition (w.l.o.g. denote that group by $H_1$).  Suppose that
$\bigcup_{j=1}^{i} G(p_j)$ covers $i$ groups $H_1, \ldots, H_i$ in the minimum partition, that is, $\bigcup_{j=1}^i H_j \subseteq \bigcup_{j=1}^i G(p_j)$, we can show that there must be a new group $H_{i+1}$ in the minimum partition such that  $\bigcup_{j=1}^{i+1} H_j \subseteq \bigcup_{j=1}^{i+1} G(p_j)$, which gives $\sol \le \opt$.   The induction step follows from the following facts. 
\begin{enumerate}
\item $p_{i+1} \not\in \bigcup_{j=1}^{i} G(p_j)$. 

\item $\ball(p_{i+1}, \alpha) \subseteq \bigcup_{j=1}^{i+1} G(p_j)$.

\item The diameter of each group in the minimum partition is at most $\alpha$. 
\end{enumerate}
Indeed, by (1) and the induction hypothesis we have $p_{i+1} \not\in \bigcup_{j=1}^{i} H_j$.  Let $H_{i+1}$ be the group containing $p_{i+1}$ in the minimum partition.  Then by (2) and (3) we must have $H_{i+1} \subseteq \ball(p_{i+1}, \alpha) \subseteq \bigcup_{j=1}^{i+1} G(p_j)$.

We next show $\opt \le O(\sol)$. This is not obvious since the diameter of a group in the greedy partition may be larger than $\alpha$ (but is at most $2\alpha$), while groups in the minimum partition have diameter at most $\alpha$.  However, in constant dimensional Euclidean space, each group in a greedy group partition can intersect at most $O(1)$ groups in the minimum cardinality partition. We thus still have $\opt \le O(\sol)$.
\end{proof}

Now we are ready to prove the theorem.

\begin{proof} [(for Theorem~\ref{thm:ambiguity})]
We can think the group partition in Algorithm~\ref{alg:IW} as a greedy process.  Let $(q_1, \ldots, q_z)$ be the sequence of points that are included in $\Sa$, according to their arriving orders in the stream.  We can generate a greedy group partition on $\bigcup_{i=1}^z \ball(q_i, \alpha)$ as follows: for $i = 1, \ldots, z$, create a new group $G(q_i) \gets \ball(q_i, \alpha) \cap S$ and update $S \gets S \backslash G(q_i)$.  
Let $\Gsub = \{G(q_1), \ldots, G(q_z)\}$. We then apply the greedy partition process on the remaining points in $S$, again according to their arriving orders in the stream. Let $q_{z+1}, \ldots, q_{\sol}$ be the representative points of the remaining groups. Let $\G = \{G(q_1), \ldots, G(q_{\sol})\}$ be the final group partition of $S$. We have the following facts.
\begin{enumerate}
\item  Each group in $\G$ intersects $\Theta(1)$ grid cell in $\Grid$.

\item Each grid cell in $\Grid$ is sampled by the hash function $h_R$ with equal probability.

\item $q_1, \ldots, q_z$ are the representative points of their groups in $\Gsub$.

\item Algorithm~\ref{alg:IW} returns a sample randomly from $q_1, \ldots, q_z$.
\end{enumerate}
By items $1$ and $2$, we know that each group in $\G$ is included in $\Gsub$ with probability $\Theta(\abs{\Gsub}/\abs{\G})$.  By items $3$ and $4$, we know that Algorithm~\ref{alg:IW} returns a random group from $\Gsub$. Therefore each group $G \in \G$ is sampled by Algorithm~\ref{alg:IW} with probability $\Theta(1/\sol) = \Theta(1/\opt)$, where the last equation is due to Lemma~\ref{lem:greedy}.  

Now for any $p \in S$, according to the greedy process and Algorithm~\ref{alg:IW}, there must be some $q \in S$ such that $G(p) \subseteq \ball(q, \alpha)$, {\em and} if $G(p)$ is sampled then $q$ is the sampled point.  So the probability that $q$ is sampled is at least the probability that $G(p)$ is sampled.  Finally, note that if $p \in \ball(q, \alpha)$ then we also have $q \in \ball(p, \alpha)$.  We thus have 
\begin{equation}
\label{eq:e-1}
\Pr[\exists\ q \in \ball(p, \alpha) \text{ s.t. $q$ is sampled}] = \Omega(1/{\opt}).
\end{equation}

On the other hand,  in constant dimensional Euclidean space $\ball(p, \alpha)$ can only intersect $O(1)$ groups in the greedy partition. We thus also have
\begin{equation}
\label{eq:e-2}
\Pr[\exists\ q \in \ball(p, \alpha) \text{ s.t. $q$ is sampled}] = O(1/{\opt}).
\end{equation}

The theorem follows from (\ref{eq:e-1}) and (\ref{eq:e-2}).  
\end{proof}

It is easy to see that the above arguments can also be applied to the sliding window case with respect to Algorithm~\ref{alg:SW}.

\begin{corollary}
For a general dataset in constant dimensional Euclidean space, there exists a sliding window algorithm (Algorithm~\ref{alg:SW}) that with probability $1 - 1/m$,  at any time step,  it outputs a point $q$ such that $\forall p \in S,\ \Pr[q \in \ball(p, \alpha)] = \Theta(1/{\opt})$, where $S$ is the set of all the points in the sliding window, and $\opt$ is the size of the minimum cardinality partition of $S$ with group radius $\alpha$.
\end{corollary}

\section{High Dimensions}
\label{sec:highD}

In this section we consider datasets in $d$-dimensional Euclidean space for general $d$.  We show that Algorithm~\ref{alg:IW}, with some small modifications, can handle $(\alpha, \beta)$-sparse dataset in $d$-dimensional Euclidean space with $\beta > d^{1.5} \alpha$ as well. 


\begin{theorem}
\label{thm:highD}
In the $d$-dimensional Euclidean space, for an $(\alpha, \beta)$-sparse dataset with $\beta > d^{1.5} \alpha$, there is a streaming algorithm such that with probability $1 - 1/m$, at any time step, it outputs a robust $\ell_0$-sample. The algorithm uses $O(d \log m)$ words of space and $O(d \log m)$ processing time per item.
\end{theorem}

\begin{remark}
We can use Johnson-Lindenstrauss dimension reduction to weaken the sparsity assumption to $\beta \ge c_\alpha \log^{1.5} m \cdot \alpha$ for some large enough constant $c_\alpha$.
\end{remark}

We place a random grid $\Grid$ with side length $d \alpha$.  Since the dataset is $(\alpha, \beta)$-sparse with $\beta > d^{1.5} \alpha$, each grid cell can intersect at most one group.  However, in the $d$-dimensional space a group can intersect $2^d$ grid cells in the worst case, which may cause difficulty to maintain $\Sr$ in small space -- in the worst case we would have $\abs{\Sr} = \Omega(2^d)$ while $\abs{\Sa}$ is still small.  Fortunately, in the following lemma we show that for any $p \in \Sf$, the probability that $p \in \Sr$ will not be too large compared with the probability that $p \in \Sa$.

\begin{lemma}
\label{lem:small-Sr}
For any fixed $p \in \Sf$, we have
$$ \Pr[p \in \Sr] \le \kappa_1 \cdot \Pr[p \in \Sa \cup \Sr],
$$
where $\kappa_1 \in (0, 1)$ is a constant.
\end{lemma}

\begin{proof}
For a group $G$, let $\ball(G, \alpha) = \{p\ |\ d(p, G)\le \alpha\}$ where $d(p, G) = \min_{q \in G} d(p, q)$.  It is easy to see that $\ball(G, \alpha)$ has a diameter of at most $3\alpha$ because the diameter of $G$ is at most $\alpha$.  

Since the random grid has side length $d \alpha$, the probability that $\ball(G, \alpha)$ is cut by the boundaries of cells in each dimension is at most 
$\mu = \frac{3}{d}$. If $\ball(G, \alpha)$ is cut by $i$ dimensions, the number of cells it intersects is at most $2^i$, and consequently $\abs{\adj(p)} \le 2^i$ for each $p \in G$.  

Recall that each cell is sampled with probability $\frac{1}{R}$, we thus have
\begin{eqnarray*}
&& \Pr[p \in \Sr \cup \Sa] \\
&\le& \sum_{i \ge 1} \Pr[p \in \Sr \cup \Sr\ |\ |\adj(p)| = i] \cdot \Pr[|\adj(p)| = i] \\
&\le& \sum_{i=0}^d {d \choose i} \mu^i (1-\mu)^{d-i} \frac{2^i}{R} \\
&=& \frac{(2\mu + 1 - \mu)^d}{R} \\
&\le& \frac{(1+\frac{3}{d})^d}{R} 
= O\left(\frac{1}{R}\right).
\end{eqnarray*}
Since $\Sa \cap \Sr = \emptyset$, we have 
\begin{eqnarray*}
\Pr[p \in \Sr] &=& \Pr[p \in \Sr \cup \Sa] - \Pr[p \in \Sa] \\
&\le&  \kappa_1 \cdot \Pr[p \in \Sa \cup \Sr]
\end{eqnarray*}
for some constant $\kappa_1 \in (0,1)$.
\end{proof}

By Lemma~\ref{lem:small-Sr}, and basically the same analysis as that in Lemma~\ref{lem:Sr}, we can bound the space usage of Algorithm~\ref{alg:IW} by $O(d \log m)$ ($O(\log m)$ points in the $d$-dimensional space) throughout the execution of the algorithm with probability $(1 - 1/m)$, and consequently the running time.

We have a similar result for the sliding window case.
\begin{corollary}
\label{thm:highD}
In the $d$-dimensional Euclidean space, for an $(\alpha, \beta)$-sparse dataset with $\beta > d^{1.5} \alpha$, there is a sliding window algorithm such that with probability $1 - 1/m$, at any time step, it outputs a robust $\ell_0$-sample. The algorithm uses $O(d \log w \log m)$ words of space and $O(d \log w \log m)$ processing time per item, where $w$ is the size of the sliding window.
\end{corollary}

\section{Distinct Elements}
\label{sec:F0}

In this section we show that our algorithms for robust $\ell_0$-sampling can be used for approximating the number of robust distinct elements in both infinite window and sliding window settings.

\paragraph{Estimating $F_0$ in the infinite window}
We first recall the algorithm for estimating the number of distinct elements on noisefree datasets by Bar-Yossef et al.~\cite{BJKST02}. We maintain an integer $z$ initialized to be $0$, and a set $B$ consisting of at most $\kappa_B / \eps^2$ items for a large enough constant $\kappa_B$.  For each arriving point, we perform an $\ell_0$-sampling with probability ${1}/{2^z}$ and add the point to $B$ if sampled.  At the time $B > \kappa_B / \eps^2$, we update $z \gets z+1$, and re-sample each point in $B$ with probability $1/2$ so that the overall sample probability is again ${1}/{2^z}$.  It was shown in \cite{BJKST02} that with probability at least $0.9$, $\abs{B} 2^z$ approximates the robust $F_0$ up to a factor of $(1+\eps)$. We can run $\Theta(\log m)$ independent copies of above procedure to boost the success probability to $1 - 1/m$.


Now we can directly plug our $\ell_0$-sampling algorithm into the framework of \cite{BJKST02}.  We simply replace the threshold $\kappa_0 \log m$ at Line~\ref{line:a-5} of Algorithm~\ref{alg:IW} with $\kappa_B/ \eps^2$.  At the time of query we return $\abs{\Sa} \cdot R$ as the estimation of the number of distinct groups. Again, running $\Theta(\log m)$ independent copies of the algorithm and taking the median will boost the constant success probability to high probability. 

%

\paragraph{Estimating $F_0$ in the sliding windows}
We can use the ideas from Flajolet and Martin \cite{FM85} (which is known as the FM sketch). We run $\Theta(1/\eps^2)$ independent copies of our sliding window algorithm. For each of the copies, we find the largest $\ell$ that $\Sa_\ell$ includes at least one non-expired sample. We average all those $\ell$'s as $\bar{\ell}$. It follows \cite{FM85} that with probability $0.9$, $\phi 2^{\bar{\ell}}$ gives $(1 + \eps)$-approximation to the robust $F_0$ 
in the sliding window, where $\phi$ is a universal constant to correct the bias (see, e.g., \cite{FM85}). We can then run $\Theta(\log m)$ copies of above procedure and taking the median to boost the success probability.  Similarly we can also plug-in our algorithm to HyperLogLog~\cite{FFGM07}.


\newcommand{\randfive}{{\tt Rand5}}
\newcommand{\randfivepl}{{\tt Rand5-pl}}
\newcommand{\randtt}{{\tt Rand20}}
\newcommand{\randttpl}{{\tt Rand20-pl}}
\newcommand{\yacht}{{\tt Yacht}}
\newcommand{\yachtpl}{{\tt Yacht-pl}}
\newcommand{\seeds}{{\tt Seeds}}
\newcommand{\seedspl}{{\tt Seeds-pl}}
\newcommand{\bx}{\mathbf{x}}
\newcommand{\bz}{\mathbf{z}}
\newcommand{\by}{\mathbf{y}}
\newcommand{\pTime}{\texttt{pTime}}
\newcommand{\pSpace}{\texttt{pSpace}}
\newcommand{\stdDevNm}{\texttt{stdDevNm}}
\newcommand{\maxDevNm}{\texttt{maxDevNm}}
\newcommand{\searchAdj}{\textsc{SearchAdj}}
\newcommand{\nruns}{\text{\#runs}}
\section{Experiments}
\label{sec:exp}
In this section, we present our experimental results for Algorithm \ref{alg:IW} as well as some implementation details. The algorithm is easy to implement yet very efficient, and can be of interest to practitioners.

\subsection{The Setup}
\paragraph{Datasets} We verify our algorithms using the following real and synthetic datasets. 
\begin{itemize}
\item \randfive: $500$ randomly generated points  in $\mathbb{R}^5$; each coordinate is a random number from $(0, 1)$.
\item \randtt: $500$ randomly generated points in $\mathbb{R}^{20}$; each coordinate is a random number from $(0, 1)$.
\item \yacht: $308$ points taken from the UCI repository yacht hydrodynamics data set.\footnote{\url{https://archive.ics.uci.edu/ml/datasets/Yacht+Hydrodynamics}} Each point is in $\mathbb{R}^7$, and it measures the sailing yachts movement.
\item \seeds:  $210$ points taken from the UCI repository seeds data set.\footnote{\url{https://archive.ics.uci.edu/ml/datasets/seeds}} Each point is in $\mathbb{R}^8$, consisting of measurements of geometrical properties of kernels belonging to three different varieties of wheat.
\end{itemize}

We perform two types of near-duplicate generations on each dataset.   In the first transformation we generate near-duplicates as follows: We first rescale the dataset such that the minimum pairwise distance is $1$. Then for each point $\bx_i\ (i = 1, 2, \ldots, n)$, we pick a number $k_i$ uniformly at random from ${1, 2, \ldots, 100}$, and add $k_i$ near-duplicate points w.r.t. $\bx_i$, each of which is generated as follows:
\begin{enumerate}
\item Generate a vector $\bz \in \mathbb{R}^d$ such that each coordinate of $\bz$ is chosen randomly from $(0,1)$.

\item Randomly sample a number $\ell \in \left(0, \frac{1}{2 d^{1.5}}\right)$ and rescale $\bz$ to length $\ell$.  Let $\hat{\bz}$ be the resulting vector.

\item Create a near-duplicate point $\by = \bx_i + \hat{\bz}$.
\end{enumerate}
Note that each point $\bx_i$, together with the near-duplicate points around it, forms a group.
We still name the resulting datasets as \randfive, \randtt, \yacht\ and \seeds\ respectively.

In the second transformation, the number of near-duplicates we generate for each data point follows the \emph{power-law} distribution. More precisely, we randomly order the points as $\bx_1, \bx_2, \ldots, \bx_n$, and for each point $\bx_i$, we add $\lceil n \cdot i^{-1} \rceil$ noisy points in the same way as above.
We denote the resulting datasets as \randfivepl, \randttpl, \yachtpl\ and \seedspl\ respectively.

\paragraph{Measurements}
For each dataset, we run each of our proposed sampling algorithms a large number of times (denoted by $\nruns$, which ranges from $200,000$ to $500,000$), and count the number of times each group being sampled. We report the following results measuring the performance of our proposed algorithm.  
\begin{itemize}
\item \pTime: Processing time per item; measured by millisecond. The running time is tested using single thread.
\item \pSpace: Peak space usage throughout the streaming process; measured by word. 
\end{itemize}
We record \pTime\ and \pSpace\ by taking the average of 100 runs where in each run we scan the whole data stream.

The following two accuracy measurements for the $\ell_0$-sampling algorithms follow from \cite{CF14}.
\begin{itemize}
\item \stdDevNm: Let $F_0$ be the number of groups, and let $f^* = \frac{1}{F_0}$ be the target probability. We calculate the standard deviation of the empirical sampling distribution and normalize it by $f^*$.
\item \maxDevNm: We calculate the normalized maximum deviation of the empirical sampling distribution as 
$$\textstyle \max_{i}\left\{\frac{|f_i - f^*|}{f^*}\right\},$$ where $f_i$ is the empirical sampling probability of the $i$-th group.
\end{itemize}

We will also visualize the number of times each group being sampled.  

All of the eight datasets are randomly shuffled before being fed into our algorithms as data streams. We return the sample at the end of a data stream.  

\paragraph{Computational Environment}
Our algorithms are implemented in C++ and the visualization code is implemented in python+matplotlib. We run our experiments in a PowerEdge R730 server equipped with 2 x Intel Xeon E5-2667 v3 3.2GHz. This server has 8-core/16-thread per CPU, 192GB Memeory and 1.6TB SSD.

\subsection{Computing $\adj(p)$ in $\mathbb{R}^d$}
Before presenting the experimental results, we discuss some important details of our implementation. 

Recall the definition of $\adj(p)$ introduced in Section~\ref{sec:2D}:
$$\adj(p) = \{C \in \Grid\ |\ d(p, C) \le \alpha\}.$$
In Section~\ref{sec:2D} we did not spell out how to compute $\adj(p)$ in $\mathbb{R}^d$ because our discussion focused on the case $d = \Theta(1)$, where computing $\adj(p)$ only takes $O(1)$ time.  However, the naive implementation that we enumerate all the adjacent cells of $\cell(p)$ and test for each cell $C$ whether we have $d(p, C) \le \alpha$ takes $\Theta(d\cdot 3^d)$ time: we have $3^d$ cells to exam, and each cell has $d$ coordinates. This is expensive for large $d$. In the following we illustrate how to compute $\adj(p)$ efficiently in practice.

According to Lemma \ref{lem:small-Sr}, the expected size of $\adj(p)$ is always bounded by a small constant given a sufficiently large separation ratio. Therefore it is possible to compute $\adj(p)$ more efficiently if we can avoid the exhaustive enumeration of the adjacent cells of $\cell(p)$. We approach this by effectively pruning cells that are impossible to meet the requirement that $d(p, C) \le \alpha$.

To simplify the presentation, we set the side length of the grid $\mathbb{G}$ to be $1$ and rescale the value $\alpha$ accordingly. We further identify the grid with the Cartesian coordinates with the origin $\mathbf{0} = (0, \ldots, 0)$.   Consider a point $p = (x_1, x_2, \ldots, x_d)$. To calculate the distance from $p$ to a cell $C$ that is adjacent to $\cell(p)$, we move $p$ to the nearest point in $C$ and record the distance being moved. The movement can be done sequentially: first in the direction of $x_1$, and then in the direction of $x_2$, $\ldots$ until $x_d$.

Figure \ref{fig:nearest-points} gives an illustration of the nearest points of $p$  in $\mathbb{R}^2$. To enumerate all those nearest points, we iterate the coordinates of $p$ from $x_1$ to $x_d$. For $x_i$, we have three options: (1) move to $\lfloor x_i \rfloor$; (2) move to $\lceil x_i \rceil$; and (3) do not move. Consequently we have $3^d$ different nearest points, and thus $3^d$ different cells which are  the cells adjacent to $\cell(p)$ including $\cell(p)$ itself. 
We then perform the enumeration using a DFS search and prune the search as long as the accumulated distance exceeds $\alpha$.  The details of this procedure are presented in Algorithm \ref{alg:searchAdj} and Algorithm \ref{alg:adj}.  

\begin{figure}
  \centering
  \includegraphics[width=0.28\textwidth]{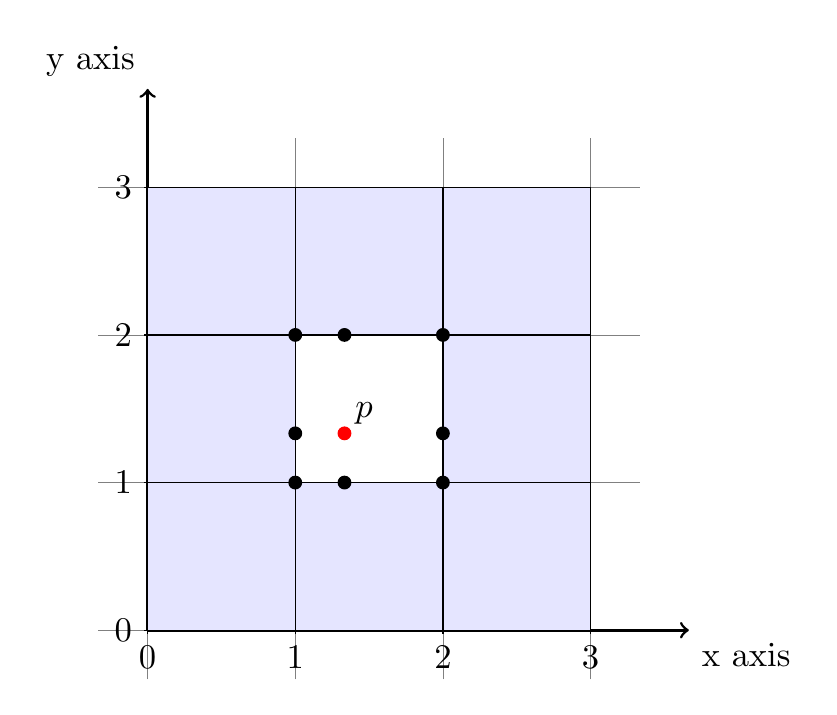}
  \caption{Nearest points of $p = (x_1, x_2) \in \mathbb{R}^2$ (red point). The eight light blue cells are cells adjacent to $\cell(p)$. There are eight black points, each of which is the nearest point from $p$ in the corresponding cell} 
  \label{fig:nearest-points}
\end{figure}

\begin{figure}[t]
     \centering
     \includegraphics[height=0.17\textwidth,width=0.22\textwidth]{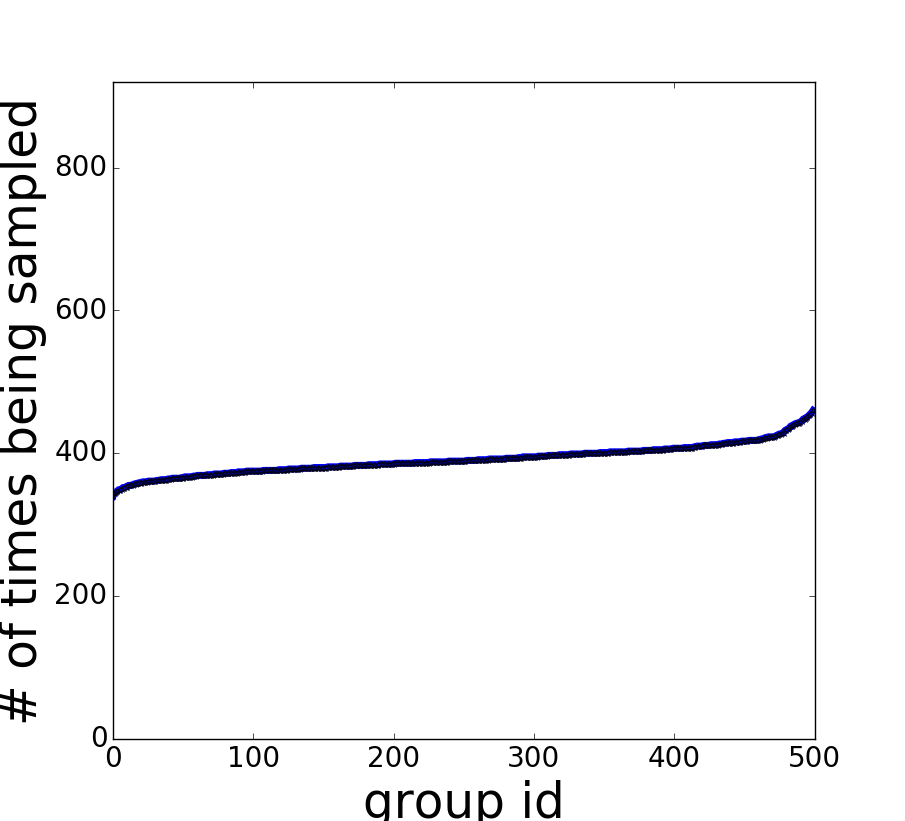}
     \caption{\randfive\ dataset. $\nruns=200,000$}
     \label{fig:rand-5}
\end{figure}

\begin{figure}[t]
     \centering
     \includegraphics[height=0.17\textwidth,width=0.22\textwidth]{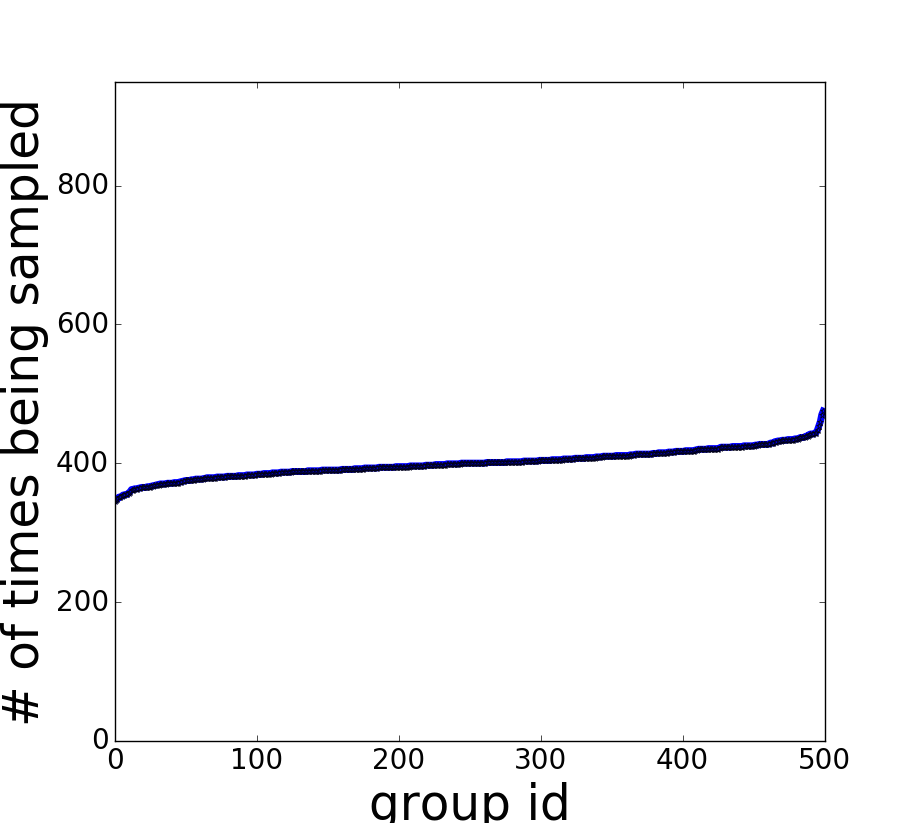}
     \caption{\randtt\ dataset. $\nruns=200,000$}
     \label{fig:randtt}
\end{figure}

\begin{algorithm}[t]
  \DontPrintSemicolon
  \tcc{$p = (x_1, \ldots, x_d) \in \mathbb{R}^d$;
    $i = 1, \ldots, d$, the depth of the DFS;
    $s$, the square distance of the movement
  }
  \If{$s > \alpha^2$}{
    \tcc{the distance of the movement exceeds $\alpha$}
    \Return\;
  }
  \If{$i > d$}{
    $q \gets (y_1, \ldots, y_d)$\;
    \tcc{Since $q$ is on the boundary, we add $0.01\cdot (q - p)$ to make sure that it moves inside a cell so that $\cell(q')$ is well defined
    }
    $q' \gets q + 0.01\cdot (q - p)$\;
    Emit $\cell(q')$\;
    \Return\;
  }
  \tcc{move $x_i$ to $\lfloor x_i \rfloor$}
  \searchAdj\parbox[t]{.6\linewidth}{
    $(p, i+1, s + (\lfloor x_i \rfloor - x_i)^2,$\\
    $(y_1,\ldots, y_{i-1}, \lfloor x_i \rfloor, \bot, \ldots) )$\;}
  \tcc{no movement}
  \searchAdj\parbox[t]{.6\linewidth}{
    $(p, i+1, s,$\\
    $(y_1,\ldots, y_{i-1},x_i, \bot, \ldots) )$\;}
  \tcc{move $x_i$ to $\lceil x_i \rceil$}
  \searchAdj\parbox[t]{.6\linewidth}{
    $(p, i+1, s + (\lceil x_i \rceil - x_i)^2,$\\
    $(y_1,\ldots, y_{i-1}, \lceil x_i \rceil, \bot, \ldots) )$\;}
\caption{\searchAdj$(p, i, s, (y_1, \ldots, y_{i-1}, \bot, \ldots, \bot))$}
\label{alg:searchAdj}
\end{algorithm}

\begin{algorithm}[t]
  \DontPrintSemicolon
  $q \gets (\bot, \bot, \ldots, \bot) \in \mathbb{R}^d$\;
  \Return all cells emitted by \searchAdj$(p, 1, 0, q)$
\caption{\adj$(p)$}
\label{alg:adj}
\end{algorithm}

\setlength{\belowcaptionskip}{-11pt}

\subsection{Results and Discussions}

We visualize our experimental results in Figure~\ref{fig:rand-5}-\ref{fig:deviation}. Figure \ref{fig:ptime} and Figure \ref{fig:pspace} show the results for time and space respectively, and Figure \ref{fig:deviation} presents the deviations of the empirical sampling distributions.  Figure \ref{fig:rand-5}-\ref{fig:seeds-pl} visualize the empirical sampling distribution of each dataset.

\begin{figure}[h]
     \centering
     \includegraphics[height=0.15\textwidth,width=0.22\textwidth]{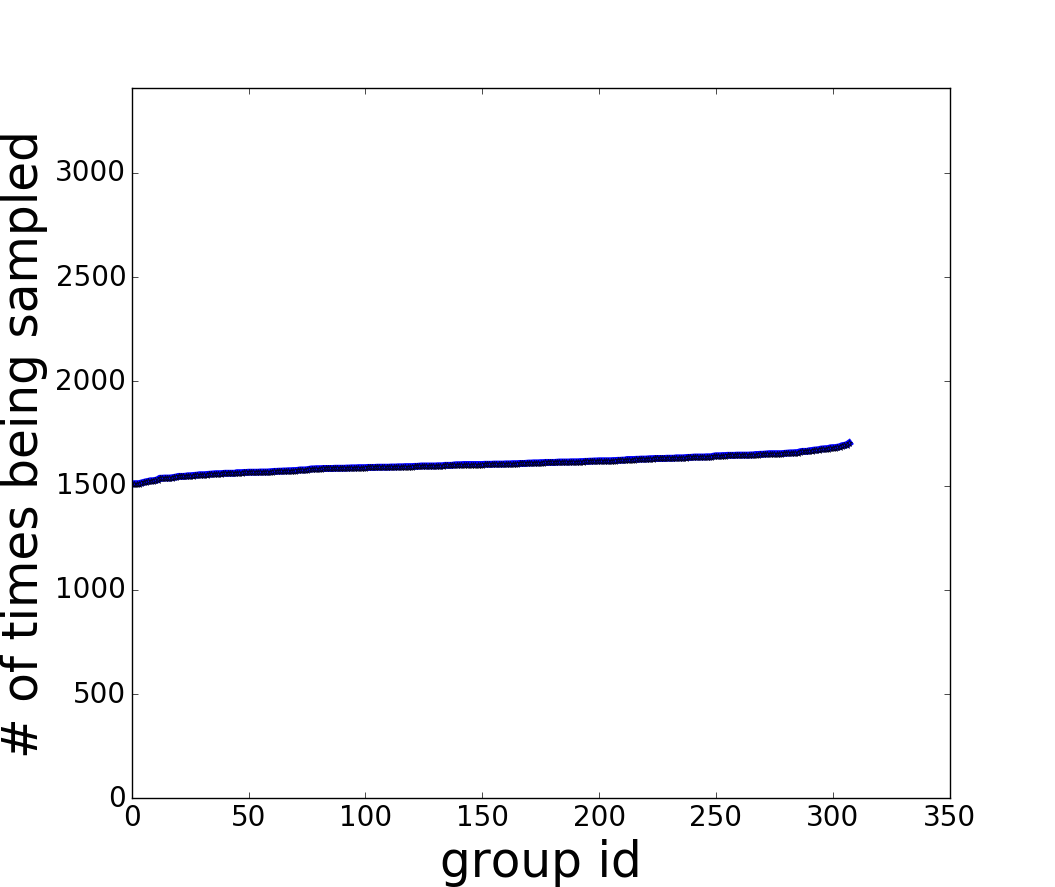}
     \caption{\yacht\ dataset. $\nruns=500,000$}
     \label{fig:yacht}
\end{figure}

\begin{figure}[h]
     \centering
     \includegraphics[height=0.15\textwidth,width=0.22\textwidth]{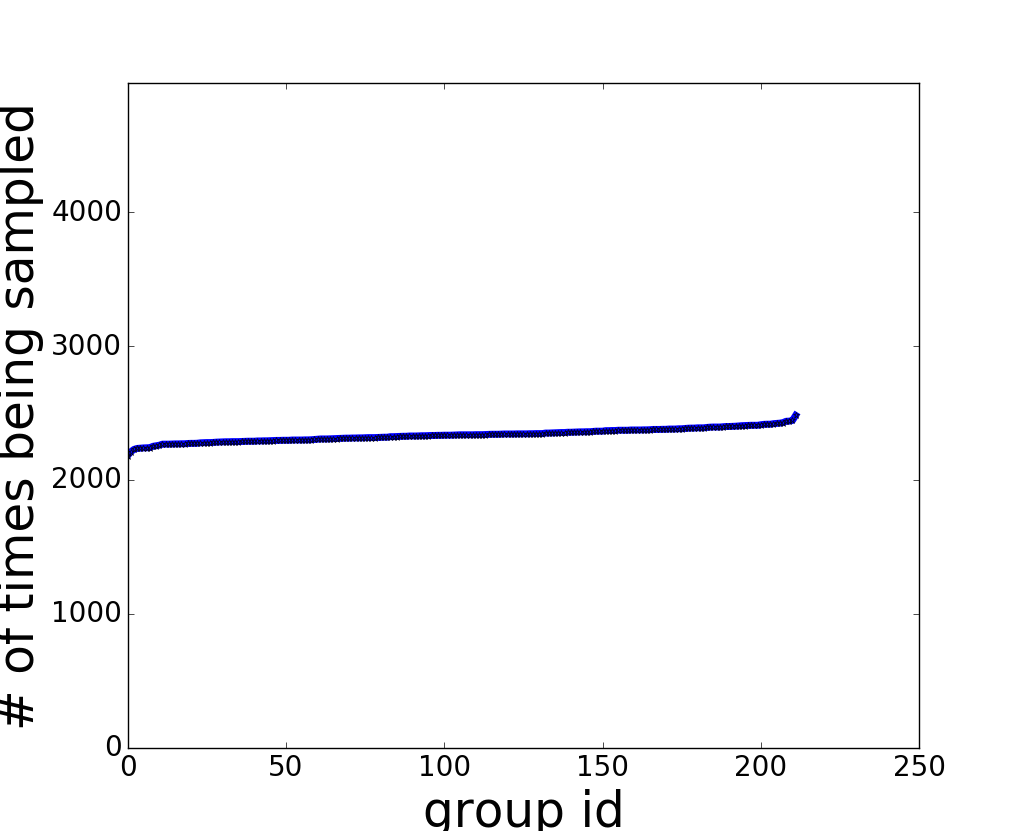}
     \caption{\seeds\ dataset. $\nruns=500,000$}
     \label{fig:seeds}
\end{figure}

\begin{figure}[h]
     \centering
     \includegraphics[height=0.15\textwidth,width=0.22\textwidth]{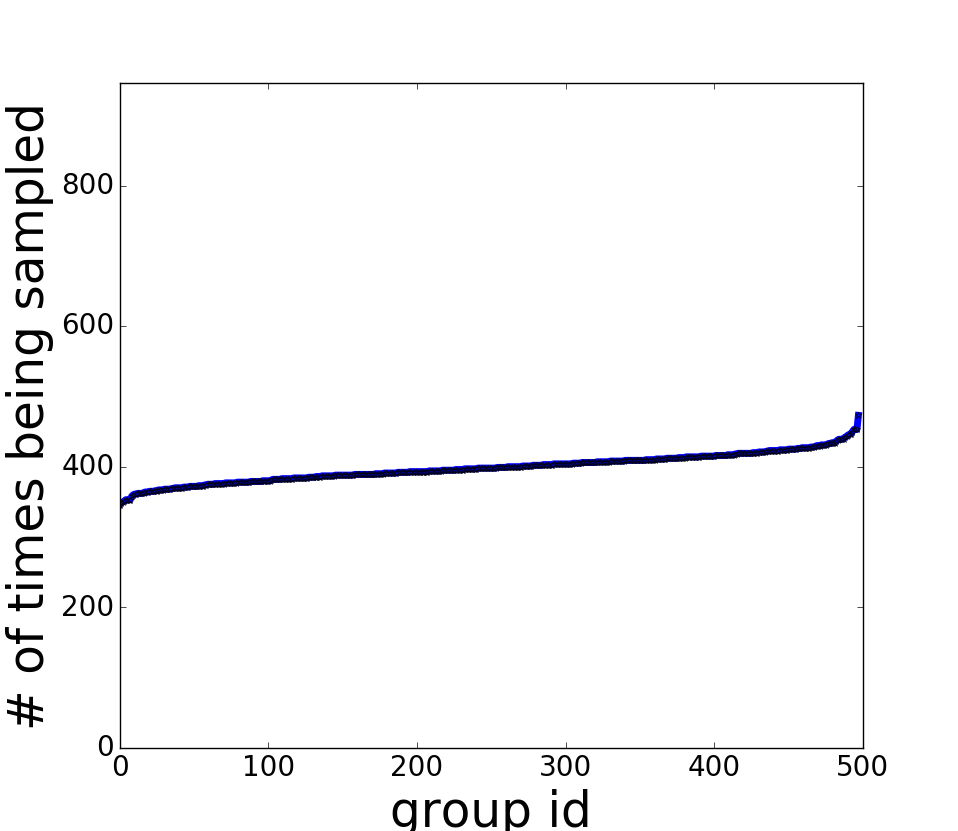}
     \caption{\randfivepl\ dataset. $\nruns=200,000$}
     \label{fig:rand-5-pl}
\end{figure}

\begin{figure}[h]
     \centering
     \includegraphics[height=0.15\textwidth,width=0.22\textwidth]{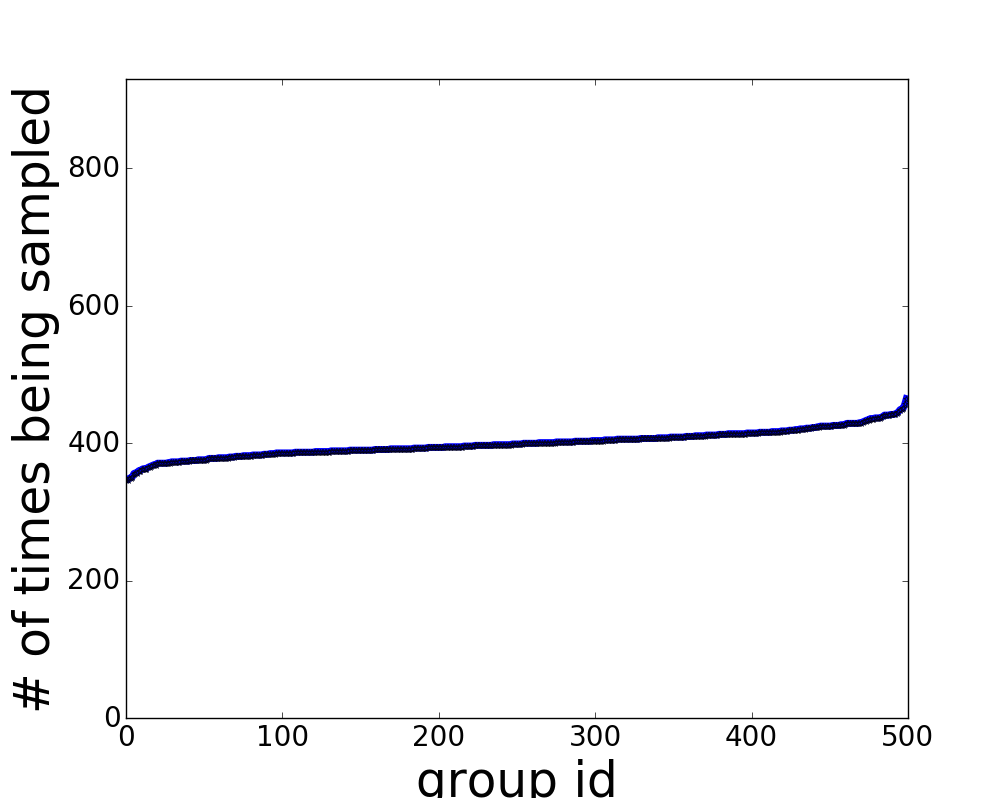}
     \caption{\randttpl\ dataset. $\nruns=200,000$}
     \label{fig:randtt-pl}
\end{figure}

\begin{figure}[h]
     \centering
     \includegraphics[height=0.15\textwidth,width=0.22\textwidth]{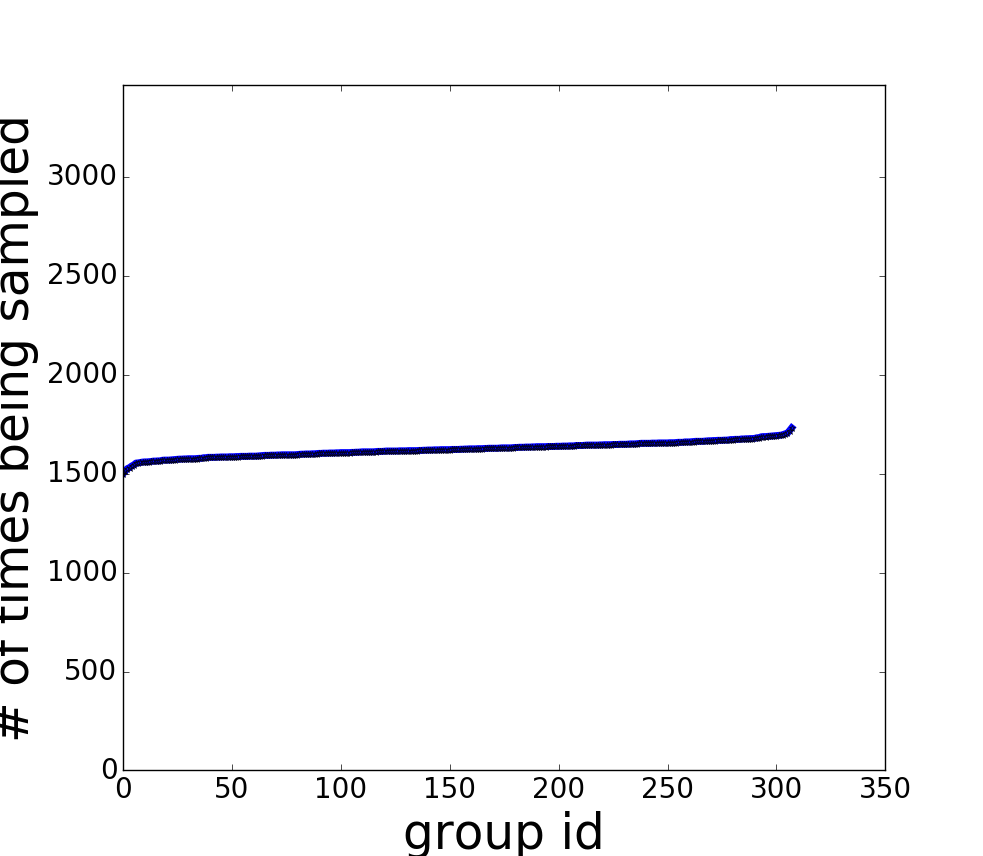}
     \caption{\yachtpl\ dataset. $\nruns=500,000$}
     \label{fig:yacht-pl}
\end{figure}

\begin{figure}[h]
     \centering
     \includegraphics[height=0.15\textwidth,width=0.22\textwidth]{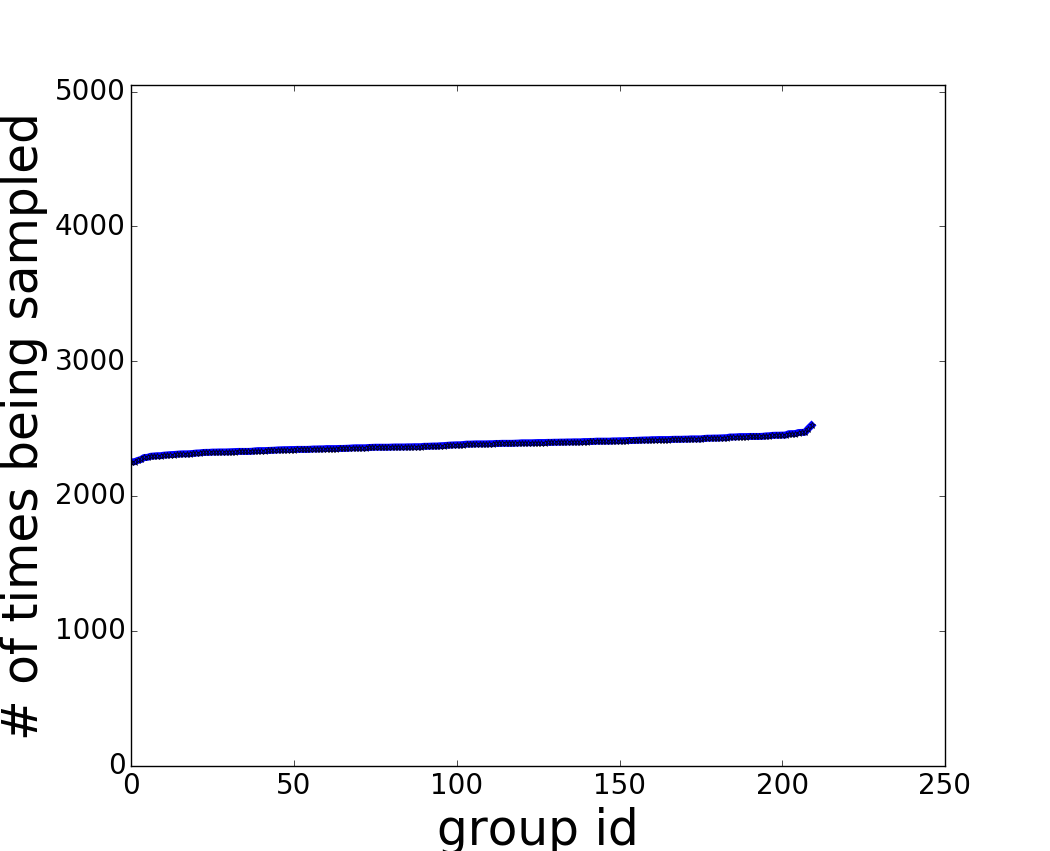}
     \caption{\seedspl\ dataset. $\nruns=500,000$}
     \label{fig:seeds-pl}
\end{figure}

We now briefly discuss these results in words.

\paragraph{Accuracy}
From Figure~\ref{fig:rand-5}-\ref{fig:seeds-pl} we can see that the empirical sampling distributions of our algorithm are very close to the uniform distribution. This can be further supported by the results presented in Figure~\ref{fig:deviation} where in all datasets, \stdDevNm\ is no larger than $0.1$ and \maxDevNm\ is no larger than $0.2$.


\paragraph{Running Time}
From Figure~\ref{fig:ptime} we can observe that Algorithm \ref{alg:IW} runs very fast. The processing time per item is only $1\sim3.5\times 10^{-5}$ second using single thread. 

By comparing the results for datasets \randfive, \randtt, \randfivepl\ and \randttpl, we observe that the running time increases when the dimension $d$ increases. This is due to the fact that manipulating vectors takes more time when $d$ increases.

\paragraph{Space Usage}
Figure~\ref{fig:pspace} demonstrates the space usage of our algorithms on different datasets. We observe that our algorithm is very space-efficient and the dimension of the data points will typically affect the space usage. 




\begin{figure}[h]
  \centering
  \includegraphics[width=0.36\textwidth]{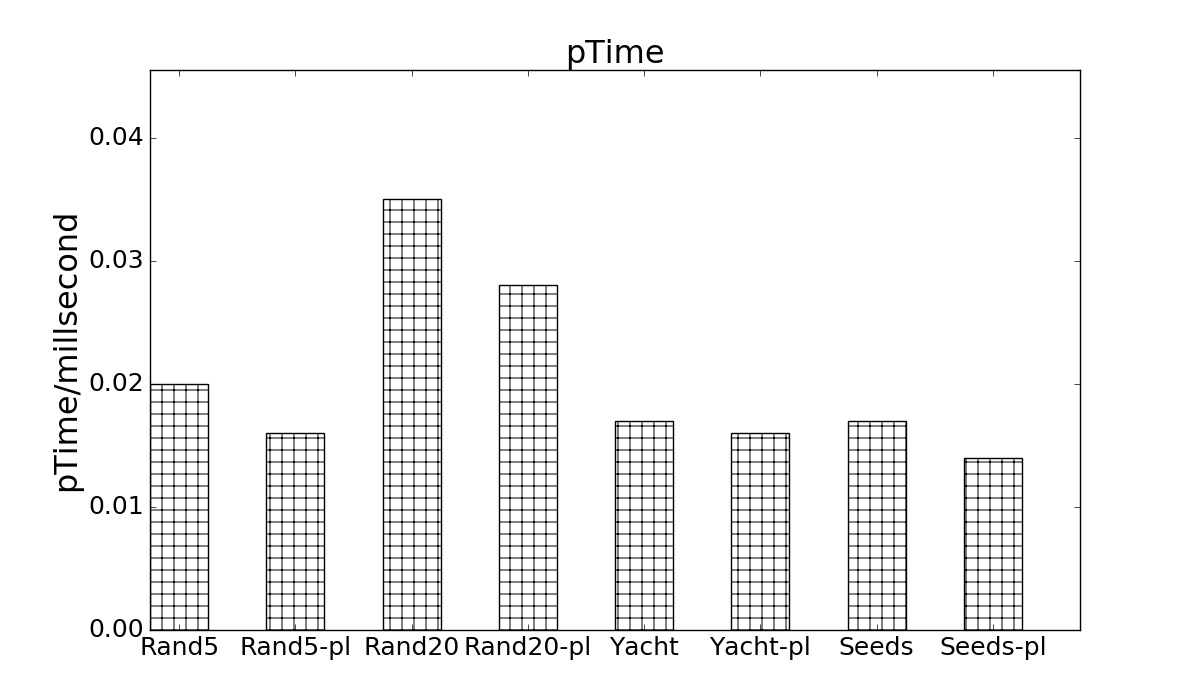}
  \caption{\pTime}
  \label{fig:ptime}
\end{figure}
\setlength{\belowcaptionskip}{-11pt}

\begin{figure}[h]
  \centering
  \includegraphics[width=0.36\textwidth]{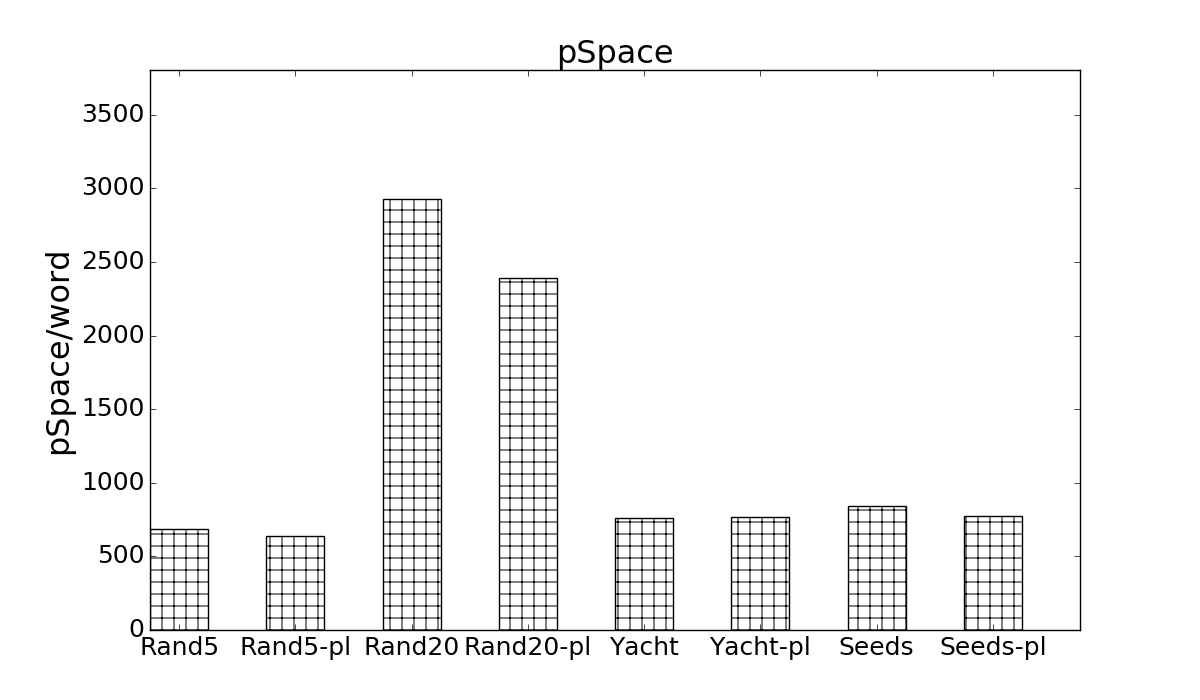}
  \caption{\pSpace}
  \label{fig:pspace}
\end{figure}

\setlength{\belowcaptionskip}{-11pt}

\begin{figure}[h]
  \centering
  \includegraphics[height=0.22\textwidth]{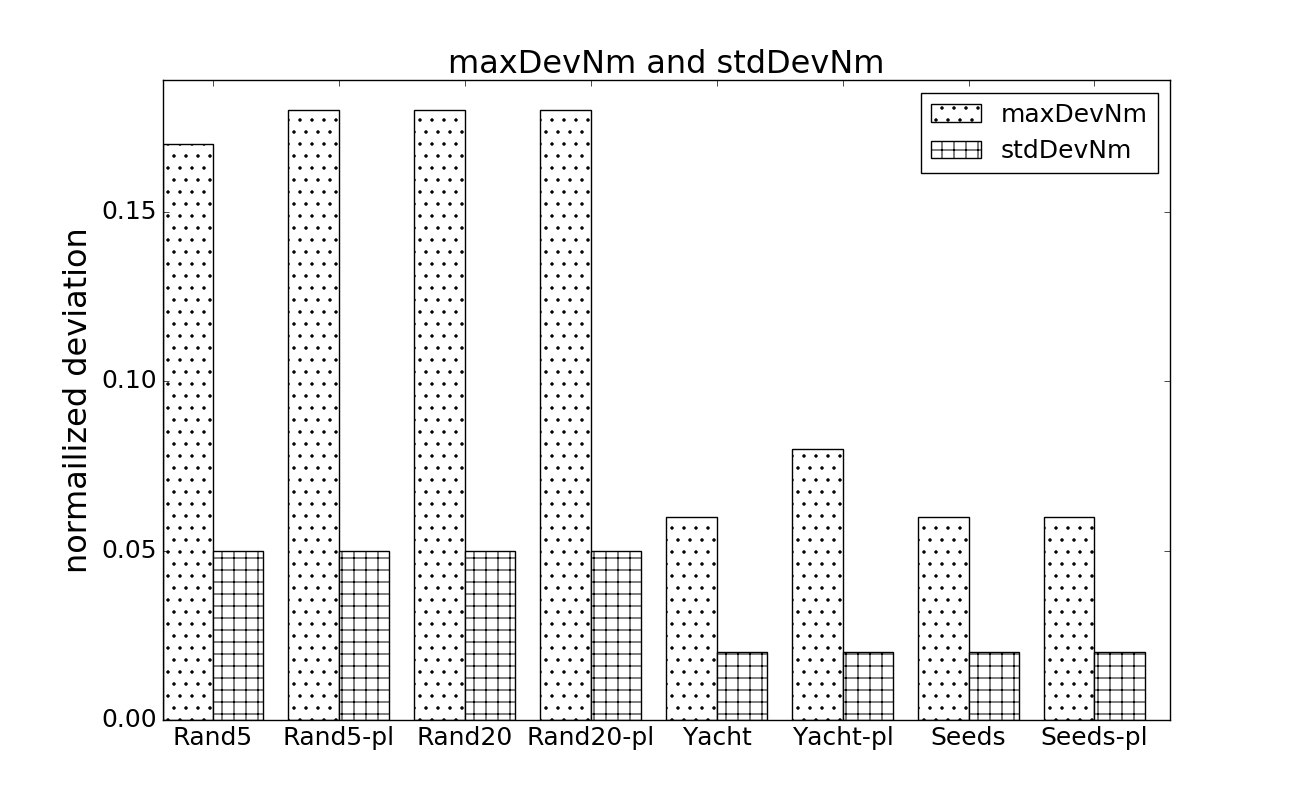}
  \caption{\maxDevNm\ and \stdDevNm}
  \label{fig:deviation}
\end{figure}

\section{Concluding Remarks}
In this paper we study how to perform distinct sampling in the noisy data stream setting, where items may have near-duplicates and we would like to treat all the near-duplicates as the same element.  We have proposed algorithms for both infinite window and sliding windows cases.  The space and time usages of our algorithms only poly-logarithmically depend on the length of the stream.  Our extensive experiments have demonstrated the effectiveness and the efficiency of our distinct sampling algorithms. 

As observed in \cite{CZ16}, the random grid we have used for dealing with data points in the Euclidean space is a particular locality-sensitive hash function,\footnote{See for example \url{https://en.wikipedia.org/wiki/Locality-sensitive_hashing}.} and it is possible to generalize our algorithms to general metric spaces that are equipped with efficient locality-sensitive hash functions. We leave this generalization as a future work.

\bibliographystyle{abbrv}
\bibliography{paper}

\end{document}